\documentclass[english,a4paper,11pt]{article}
\usepackage[T1]{fontenc}
\usepackage[latin9]{inputenc}
\usepackage{xcolor}
\usepackage{pdfcolmk}
\usepackage{babel}
\usepackage{float}
\usepackage{amsthm}
\usepackage{amsmath}
\usepackage{graphicx}
\PassOptionsToPackage{normalem}{ulem}
\usepackage{ulem}
\usepackage[unicode=true,
 bookmarks=true,bookmarksnumbered=false,bookmarksopen=false,
 breaklinks=false,pdfborder={0 0 1},backref=false,colorlinks=true]
 {hyperref}
\hypersetup{pdftitle={Hierarchical Floorplans of Order K},
 pdfauthor={Shankar Balachandran, Sajin Koroth},
 pdfsubject={Discrete Applied Mathematics},
 pdfkeywords={Floorplanning, Combinatorics, Forbidden Permutations}}

\makeatletter

\theoremstyle{plain}
\newtheorem{thm}{\protect\theoremname}
  \theoremstyle{definition}
  \newtheorem{defn}[thm]{\protect\definitionname}

\usepackage[lined,linesnumbered,boxed]{algorithm2e}

\usepackage{amsthm}

\newcommand{\HFO}{{{\hbox{HFO}}}}

\newtheorem{obs}{Observation}

\makeatother

  \providecommand{\definitionname}{Definition}
\providecommand{\theoremname}{Theorem}

\begin{document}

\title{A Study on Hierarchical Floorplans of Order $k$}

\author{Shankar Balachandran\footnote{\texttt{bshankar@cse.iitm.ac.in}},
        Sajin Koroth\footnote{\texttt{sajin@cse.iitm.ac.in}}}

\date{12th September, 2011}
\maketitle
\begin{abstract}
A floorplan is a rectangular dissection which describes the relative
placement of electronic modules on the chip. It is called a mosaic
floorplan if there are no empty rooms or cross junctions in the rectangular
dissection. We study a subclass of mosaic floorplans called hierarchical
floorplans of order $k$ (abbreviated $\HFO_{k}$). A floorplan is
a hierarchical floorplan of order $k$ if it can be obtained by starting
with a single rectangle and recursively embedding mosaic floorplans
of at most $k$ rooms inside the rooms of intermediate floorplans.
When $k=2$ this is exactly the class of slicing floorplans as the
only distinct floorplans with two rooms are a room with a vertical
slice and a room with a horizontal slice respectdeively. And embedding
such a room is equivalent to slicing the parent room vertically/horizontally.
In this paper we characterize permutations corresponding to the Abe-labeling
of $\HFO_{k}$ floorplans and also give an algorithm for identification
of such permutations in linear time for any particular $k$. We give
a recurrence relation for exact number of $\HFO_{5}$ floorplans with
$n$ rooms which can be easily extended to any $k$ also. Based on
this recurrence we provide a polynomial time algorithm to generate
the number of $\HFO_{k}$ floorplans with $n$ rooms. Considering
its application in VLSI design we also give moves on $\HFO_{k}$ family
of permutations for combinatorial optimization using simulated annealing
etc. We also explore some interesting properties of Baxter permutations
which have a bijective correspondence with mosaic floorplans.
\end{abstract}
\tableofcontents{}

\section{Introduction}

In the design of VLSI circuits floorplanning is an important phase.
The aim of the floorplanning phase is to minimize certain objective
functions like interconnection wire length while considering only
the relative placement of the blocks. During floorplanning the designers
have additional flexibility in terms of size shape and orientation
of the modules on chip. The shape of the chip and that of the modules
is usually a rectangle. A floorplan describes the relative placement
of the blocks. Hence it is modeled mathematically as a dissection
of a rectangle with axis parallel (horizontal/vertical) non-intersecting
line segments which captures the relative placement of the blocks.

\subsection{Mosaic floorplans}

Mosaic floorplans are rectangular dissections where there are no cross
junctions and the number of rooms is equal to the number of modules
to be placed on the chip. That is there are no empty rooms.

\subsection{Slicing Floorplans }

A floorplan is called a slicing floorplan if it can be obtained from
a rectangle by dissecting it recursively horizontally or vertically.

\subsection{Slicing Tree}

A slicing floorplan can be represented by a rooted tree called slicing
tree\cite{wong1986new}. A slicing tree is a \textbf{rooted binary
tree} with the following properties:
\begin{itemize}
\item Every internal node is labeled either $V$ or $H$
\item Each leaf node correspond to a basic room in the final floorplan.
\end{itemize}
A slicing tree captures the order in which the basic rectangle was
divided recursively to obtain the final floorplan. But as shown in
figure \ref{Flo:Slicing_Skewed_Slicing} there can be multiple slicing
trees corresponding to the same floorplan. To avoid this problem we
define a sub-class of slicing trees called \textbf{skewed slicing
trees} which are essentially slicing trees which also obey the following
rule:
\begin{itemize}
\item A internal node (labeled from $\left\{ V,H\right\} $) and its right
child cannot have the same label.
\end{itemize}
This rule produces a unique tree corresponding to slicing floorplan
by eliminating symmetry associated with horizontal and vertical cuts
by ensuring that always the first operation from left to right and
top to bottom is the parent node at that level. This is achieved by
the extra rule above as it says that a $V$ node has to have an $H$
node or a leaf as the right child, because if the $V$ was not the
first one from left to right then $V$ ought to have another $V$
as its right child. Similarly symmetry associated with $H$ is also
removed by skewness.

\begin{figure}
\label{Flo:Slicing_Skewed_Slicing}\caption{Slicing Trees and Skewed Slicing Trees}
\begin{center}

\includegraphics[scale=0.4]{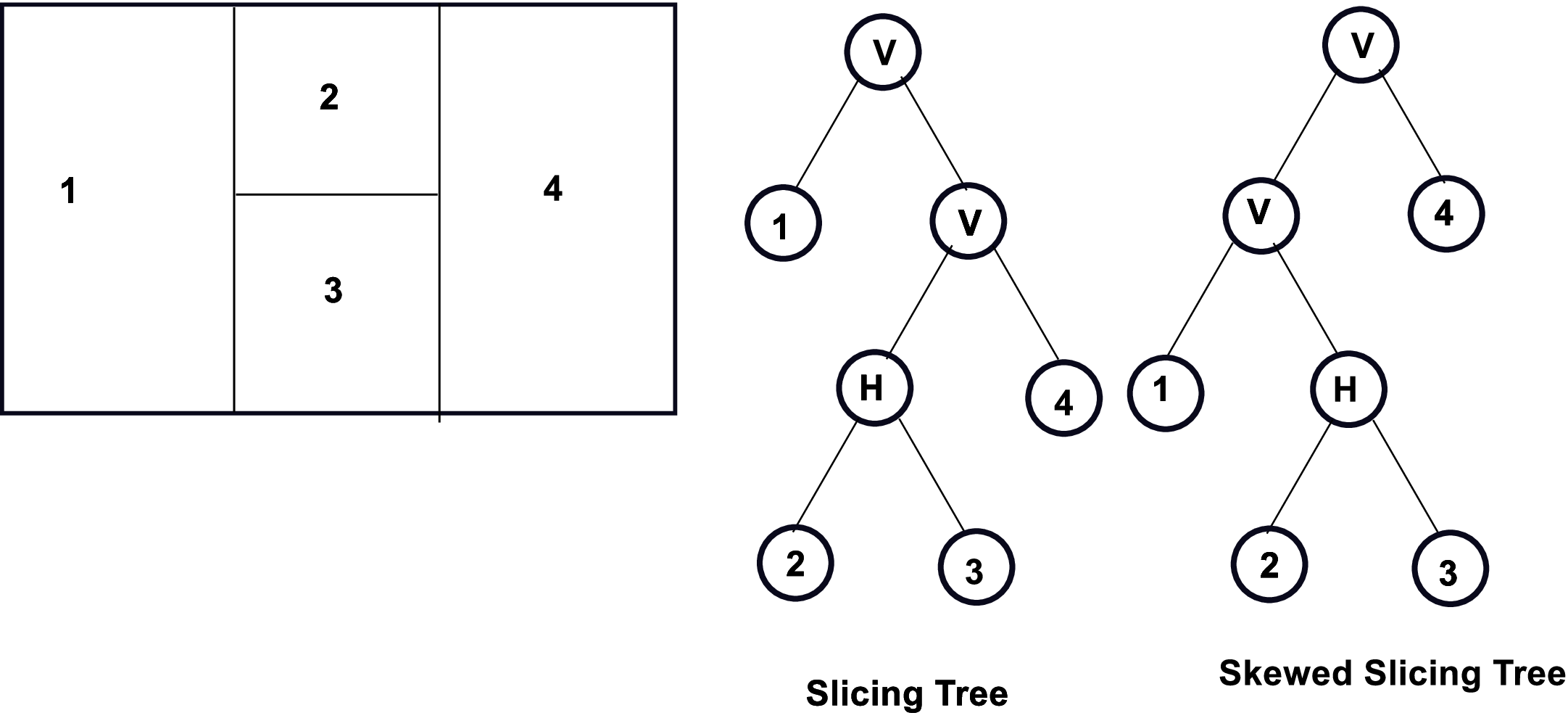}\end{center}
\end{figure}

\subsection{$\HFO_{5}$ floorplans}

A floorplan is said to be hierarchical of order 5 if it can be obtained
from a rectangle by recursively sub-dividing each rectangle into either
two parts by a horizontal or a vertical line segment or into five
parts by a wheel structure shown in figure \ref{Fig:Wheels}, the
only non-slicing mosaic floorplans with at most $5$ rooms.

\begin{figure}
\caption{Wheels}

\label{Fig:Wheels}

\begin{center}

\includegraphics[scale=0.28]{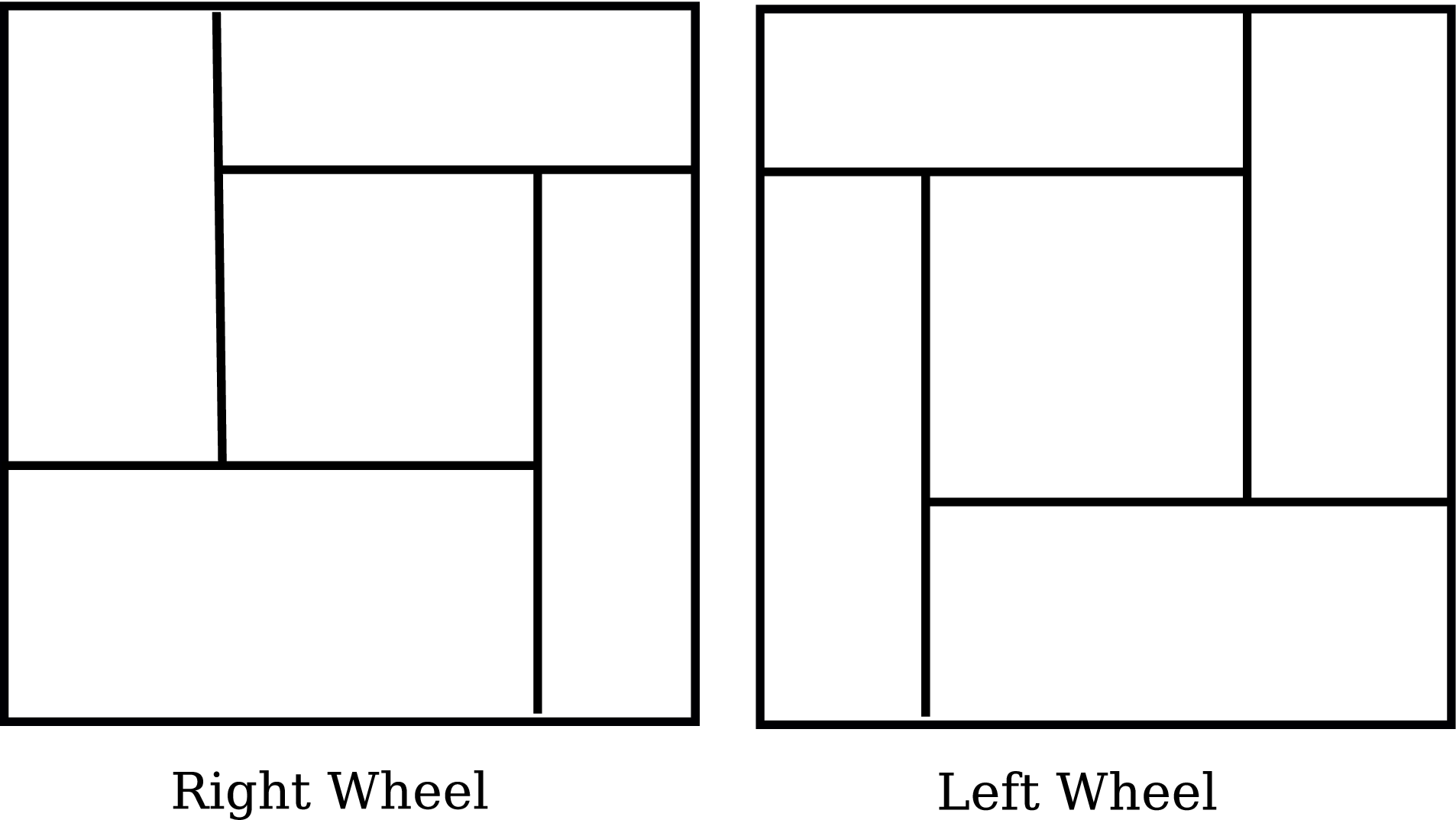}\end{center}
\end{figure}

\subsection{Hierarchical Floorplans of Order $k$}

The definition of hierarchical floorplans of order $5$ can be extended
to any $k$, by defining hierarchical floorplans of order $k$ as
all mosaic floorplans which can be obtained from a rectangle by recursively
sub-dividing each rectangle into $l$ parts ($l\leq k$) by embedding
a mosaic floorplan with $l$ rooms. When $k=2$ this becomes the class
of slicing floorplans and when $k=5$ it becomes $\HFO_{5}$.

\subsection{Pattern Matching Problem on Permutations}

Pattern matching problem for permutation is given a permutation $\pi\in S_{n}$
called \textbf{text} and another permutation $\sigma\in S_{k}$ called
\textbf{pattern} we would like to know if there exists $k$ indices
$i_{1}<i_{2}<i_{3}<i_{4}<\ldots<i_{k}$ such that the numbers $\pi[i_{1}],\pi[i_{2}],\pi[i_{3}],\pi[i_{4}],\ldots,\pi[i_{k}]$
are in the same relative order as $\sigma[1],\sigma[2],\sigma[3],\sigma[4],\ldots,\sigma[k]$,
that is $\pi[i_{h}]>\pi[i_{l}]$ if and only if $\sigma[h]>\sigma[l]$.
If $\pi$ contains such a sequence we call text $\pi$ \textbf{contains
the pattern} $\sigma$ and sub-sequence $\pi[i_{1}],\pi[i_{2}],\pi[i_{3}],\pi[i_{4}],\ldots,\pi[i_{k}]$
is said to \textbf{match the pattern} $\sigma$.

\subsection{Baxter Permutations}

A Baxter permutation on $[n]=1,2,3,\ldots,n$ is a permutation $\pi$
for which there are no four indices $1\leq i<j<k<l\leq n$ such that
\begin{enumerate}
\item $\pi[k]<\pi[i]+1=\pi[l]<\pi[j]$; or
\item $\pi[j]<\pi[i]=\pi[l]+1<\pi[k]$ 
\end{enumerate}
That is $\pi$ is a Baxter permutation if and only if whenever there
is sub-sequence matching the pattern $3142$ or $2413$ then the absolute
difference between the first and last element of the sub-sequence
is always greater than $1$. For example $2413$ is not Baxter as
the absolute difference between $2$ and $3$ is $1$ and $41352$
is Baxter even though the sub-sequence $4152$ matches the pattern
$3142$ but the absolute difference between first and last of the
sub-sequence is $|4-2|=2>1$.

\subsubsection{Algorithm FP2BP}

Eyal Ackerman et. al\cite{Ackerman20061674} in 2006 showed the existence
of a direct bijection between mosaic floorplans with $n$ rooms and
Baxter permutations of length $n$. They did this by providing two
algorithms, one which takes a mosaic floorplan and produces the corresponding
Baxter permutation and another which takes a Baxter permutation and
produces the corresponding mosaic floorplan. To explain the algorithm
we have to define the following operation on a mosaic floorplan.
\begin{defn}
[Top-Left Block deletion]Let $f$ be a mosaic floorplan with $n>1$
blocks and let $b$ be the top left block in $f$. If the bottom-right
corner of $b$ is a '$\dashv$'-(resp., '$\perp$') junction, the
one can delete $b$ from $f$ by shifting its bottom(resp., right)
edge upwards(resp., leftwards), while pulling the \textbf{T}-junctions
attached to it until the edge hits the bounding rectangle.

Similarly the block deletion operation can be defined for the other
corners of the floorplan. The algorithm for obtaining a baxter permutation
from a mosaic floorplan is the following.
\end{defn}

\begin{algorithm}

\SetKwInOut{Input}{Input}
\SetKwInOut{Output}{Output}

\Input{A mosaic floorplan $f$ with $n$ blocks}
\Output{A (Baxter) permutation of length $n$}

Label the rooms in their top-left deletion order from $\lbrace 1,\dots,n \rbrace$ \;
Obtain the permutation by arranging the room labels in their bottom-left deletion order \;

\caption{Algorithm FP2BP}
\end{algorithm}

The action of the algorithm on a mosaic floorplan is illustrated by
the figures \ref{Flo:FP2BP_Labelling} and \ref{Flo:FP2BP_Extraction}.
The permutation thus obtained is called the \textbf{Abe-label} of
the corresponding floorplan.

\begin{figure}
\label{Flo:FP2BP_Labelling}

\caption{FP2BP Labelling Phase}
\begin{center}

\includegraphics[scale=0.4]{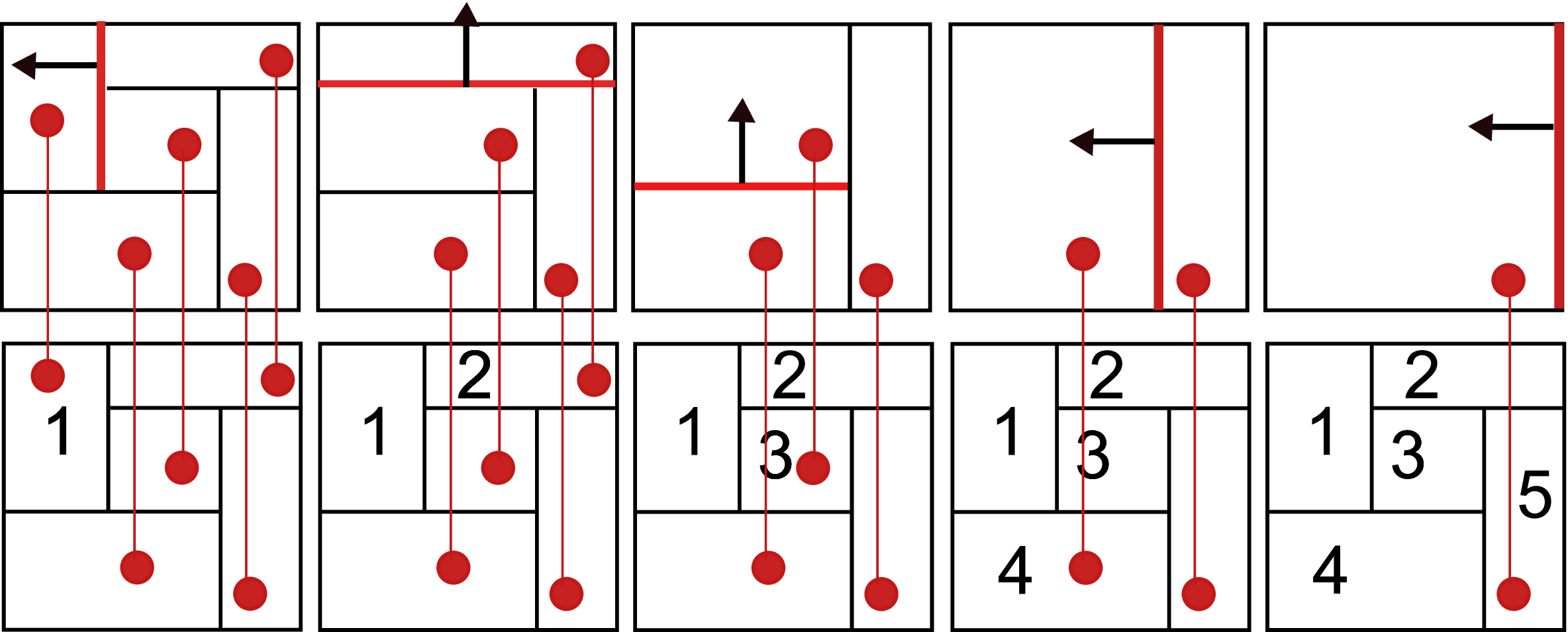}\end{center}
\end{figure}

\begin{figure}
\label{Flo:FP2BP_Extraction}\caption{FP2BP Extraction of permutation Phase}
\begin{center}

\includegraphics[scale=0.4]{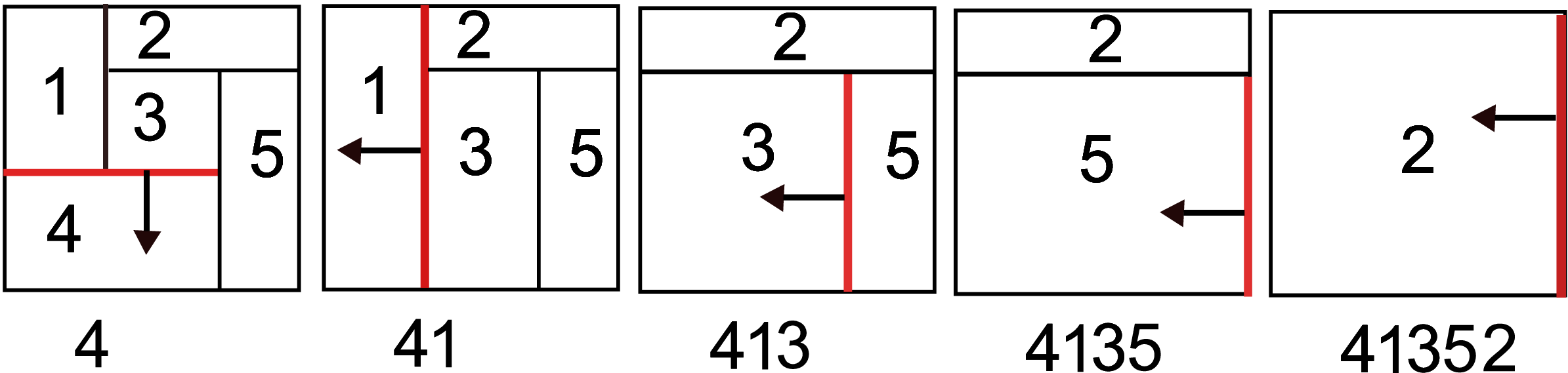}\end{center}
\end{figure}

\subsection{Simple Permutations}
\label{Defn:Block}
A \textbf{block} in a permutation is a set of consecutive positions
(called segments) which is mapped to a range of values. The trivial
block of a permutation are singleton blocks and the block $1\dots n$.
For example in the permutation $\pi=3421$ segment $1\dots3$ is a
block as $\pi$ maps $1\dots3$ to $\left\{ 2,3,4\right\} $ which
is a range but the segment $2\dots4$ is not a block as it is mapped
to $\left\{ 1,2,4\right\} $ which is not a range as $3$ is missing.
A permutation is called simple when all its block are trivial blocks.
An example of a simple permutation is $\pi=41352$. Also note that
$\pi$ above is the Abe-label of right rotating wheel. A \textbf{one-point
deletion} on a permutation $\pi$ is deletion of a single element
at some index $i$ and getting a new permutation $\pi^{'}$ on $[n-1]$
by rank ordering the remaining elements. For example one-point deletion
at index $3$ of $41352$ gives $4152$ which when rank ordered gives
the permutation $3142$.

\subsection{Block Decomposition of a permutation}

Simple permutations are an interesting class of permutations for the
reason that arbitrary permutations can be built just using simple
permutations. A \textbf{block decomposition\cite{Albert20051}} of
a permutation $\sigma$ is a partition of $\sigma$ into blocks. A
block decomposition is non-trivial if there is at least one block
which is non-trivial. Given the block decomposition of $\sigma$,
its pattern is the permutation defined by the relative order of the
blocks. For example $451362$ has the non-trivial decomposition $(45)(1)(3)(6)(2)$
with the pattern of decomposition being $41352$. We can think of
$453162$ being constructed from $41352$ by inflating each of the
elements $12,1,1,1$ and $1$ into blocks. This can be represented
as wreath product of permutations as $451362=41352\left[12,1,1,1,1\right]$.

\subsection{Exceptionally Simple Permutations}

The following simple permutations are called \textbf{exceptional :}

\begin{gather}
246\dots(2m)135\dots(2m-1)\\
(2m-1)(2m-3)\dots1(2m)(2m-2)\dots2\\
(m+1)1(m+2)2\dots(2m)m\\
m(2m)(m-1)(2m-2)\dots1(m+1)
\end{gather}

They are called exceptionally simple because no one-point deletion
of an exceptionally simple permutation can give a simple permutation.
For example $246135$ is an exceptionally simple permutation of length
$6.$ If we delete 2 from our example we get 35124 which is not simple
as the segment $3\dots4$ is mapped to $12$ hence is non-trivial
block. It can be easily verified that every one point deletion from
above permutation results in a non-trivial block. The interesting
thing about exceptionally simple permutations is that there is always
a two point deletion which yields a simple permutation of length $n-2$\cite{Albert20051}.
For example if you delete ${1,2}$ from $246135$ you get $2413$,
a simple permutation.

\section{Previous Work}

Wong and Liu\cite{wong1986new} were the first to consider the use
of stochastic search methods like simulated annealing for floorplan
optimization problem. In their seminal paper about simulated annealing
based search on the family of slicing floorplans (\cite{wong1986new}
), they introduced slicing trees and proved that there is a one-one
correspondence between slicing floorplans with $n$ rooms and skewed
slicing trees with $n$ leaves. They also proved that there is a one-one
correspondence between skewed slicing trees with $n$ leaves and normalized
polish expression of length $2n-1$. Wong and The \cite{WongAndThe}
gave a representation of hierarchical floorplans of order $5$ extending
the normalized polish expressions of slicing floorplans to incorporate
wheels which are the only non-slicing floorplans with at most five
rooms. They also described neighbourhood moves for simulated annealing
search on $\HFO_5$ floorplans based on this representation.

Sakanushi et. al \cite{Sakanushi1193019} were the first to consider
the number of distinct mosaic floorplans. They found a recursive formula
for this number. Yao et. al\cite{Yao606607} showed a bijection between
mosaic floorplans and twin binary trees whose number is known to be
the number of Baxter permutations (\cite{dulucq1998baxter}). They
have also shown that the number of distinct slicing floorplans containing
$n$ blocks is the $n-1$th Shr\"{o}der number. Later Ackerman et.
al\cite{Ackerman20061674} constructed a bijection between mosaic
floorplans with $n$-rooms to Baxter permutations on $[n]$. They
also proved that this bijection when restricted to slicing floorplans
gives a bijection from slicing floorplans with $n$-rooms to separable
permutations on $[n]$. And with this bijection a unique permutation,
corresponding to any mosaic floorplan or naturally for a floorplan
which belongs to a subclass of mosaic floorplans, can be obtained.

Simple permutations and their properties were studied first by \cite{Albert20051}.
They proved a crucial theorem about exceptionally simple permutations
using a result from a paper by \cite{Schmerl:1993:CIP:157146.157164}
about critically in-decomposable partially ordered sets.

Shen et. al\cite{shen2003bounds} presented a generating function
based approach to count skewed slicing trees, to obtain a tight bound
on number of slicing floorplans with $n$ rooms. Chung et. al\cite{Chung1978382}
obtained closed form expression for the number of Baxter permutations
of length $n$ using a generating tree based approach. 

\section{An Infinite Hierarchy}

Hierarchical floorplans form an infinite hierarchy whose levels are
$\HFO_k$ floorplans for a specific value of $k$ and it is such
that each level has at least one floorplan which is not contained
in the level below.
\begin{thm}
For any $k\geq7$, $\HFO_k\setminus\HFO_{k-1}\neq\phi$.
\end{thm}
\begin{proof}
An $\HFO_k$ floorplan which is not $\HFO_j$ for 
$j<k$ should be such that you should not be able find a proper
subset of basic rectangles in the given floorplan which are contained
in an enveloping rectangle. Because if you are able to find such a
set of rectangles then you will be able to construct this floorplan
hierarchically by starting with the floorplan obtained by removing 
all the basic rectangles having the above mentioned property leaving
only the enveloping rectangles intact and then embedding the floorplan
constituted by the basic rectangles which were contained in the enveloping
rectangle. 
We will first show that for any odd number $k\geq7$
there is a hierarchical floorplan of order $k$ which is not hierarchical
floorplan of order $j$ for any $j<k$. The proof is evident from
the geometric construction given in Figure:\ref{Flo:hfo7_to_hfo9}.
The procedure is to start with an $\HFO_7$ floorplan which is not
$\HFO_j$ for any $j<7$ and then take a line-segment which touches the bounding
box of all rectangles so that there are no parallel line-segments to
its left and then cut it half-way through as shown in the figure
and insert a \textbf{T}-junction. This way the resulting floorplan
will also have no proper subset of basic rectangles which are contained
in an enveloping rectangle, thus making it not contained in any 
lower levels of Hierarchy. The procedure increases the number of 
rooms in the floorplan by $2$. Note that the in the floorplan
obtained using the above procedure there exists a line-segment
which touches the bounding box of all rectangles so that there are
no parallel line-segments to its left. Hence it can be applied recursively to get
an $\HFO_k$ having the above mentioned properties for $k$ odd.

\begin{figure}
\caption{Constructing $\HFO_9$ from $\HFO_7$}
\label{Flo:hfo7_to_hfo9}\begin{center}

\includegraphics[scale=0.3]{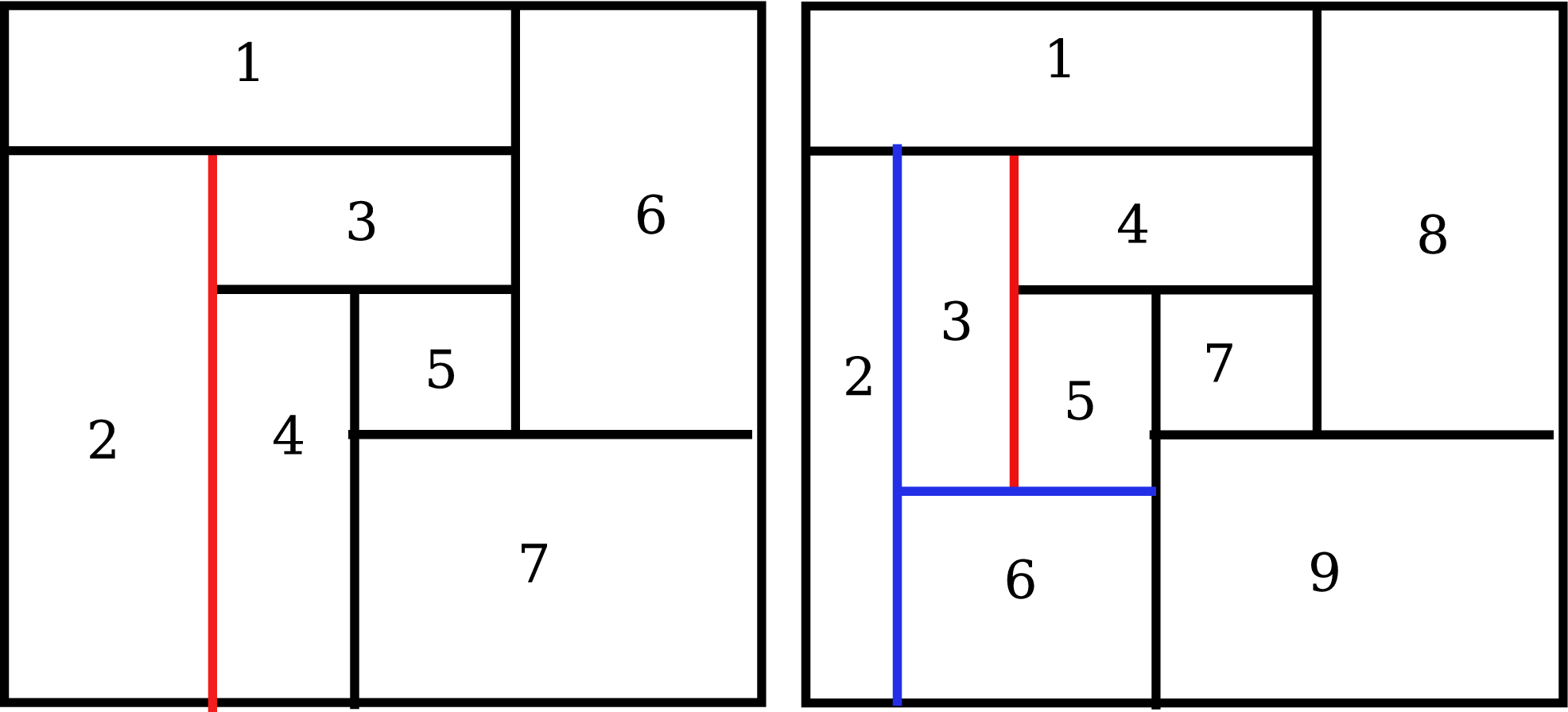}\end{center}
\caption{Constructing $\HFO_{10}$ from $\HFO_8$}
\label{Flo:hfo8_to_hfo10}\begin{center}

\includegraphics[scale=0.3]{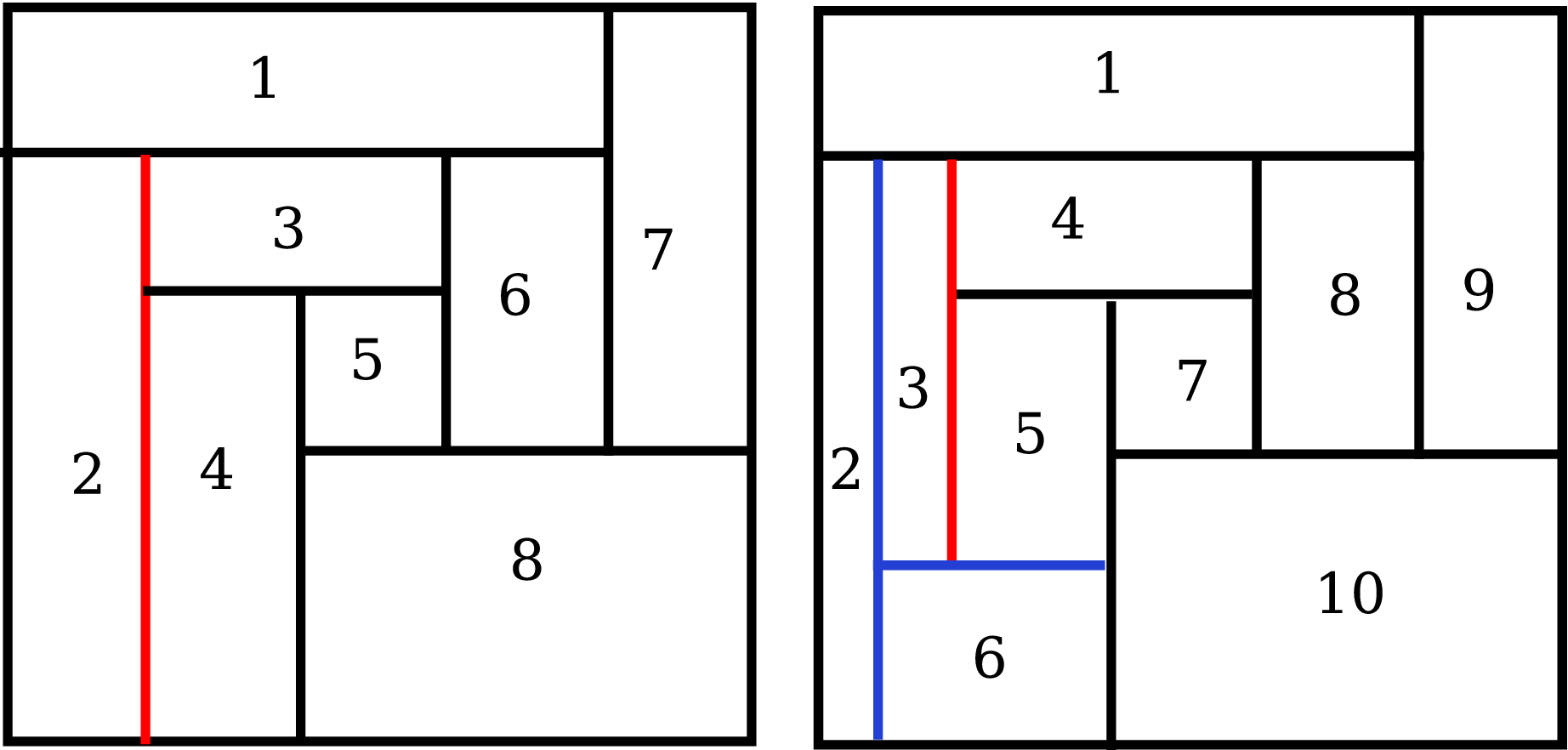}\end{center}
\end{figure}

For an even $k\geq8$ we use the same proof technique but we start
from an $\HFO_8$ which is not an $\HFO_j$ for any $j<8.$ Figure:\ref{Flo:hfo8_to_hfo10}
demonstrates the construction. The construction can applied recursively
to prove the existence of an $\HFO_k$ which not $\HFO_j$ for any $j<k$
and $k$ even.\end{proof}

\section{Uniquely $\HFO_k$}

We call an $\HFO_k$ floorplan with $k$-rooms which is not a $\HFO_j$ floorplan
for any $j<k$, Uniquely $\HFO_k$. We will prove that they are in
bijective correspondence with those Baxter permutations of length
$k$ which are simple permutations of length $k$. Given an $\HFO_k$ floorplan with $k$
  rooms if you are able to find set of $j, 1 < j < k$ basic blocks
  which are contained in an enveloping rectangle then you will be able
  to generate this floorplan in the following way.  Consider the
  floorplan obtained by replacing these $j$ basic blocks by just the
  enveloping rectangle and then place the $\HFO_j$ floorplan formed
  by these basic rectangles inside that room. Hence it is clear that
  the resulting floorplan belongs to $\HFO_{\max \lbrace k-j, j
    \rbrace}$, and since both $k-j$ and $j$ are strictly greater than
  one we get that the resulting floorplan is not Uniquely
  $\HFO_k$. So if you are not able to find a non-trivial set of
  basic rectangles which together form another rectangle in a mosaic
  floorplan of $k$ rooms then it is a Uniquely $\HFO_k$ floorplan.
  We need the following crucial observation for formal proof of the characterization.
  \begin{obs}
  \label{Obs:BlockEnvolopingRectangle}
   In the permutation $\pi$ produced by the FP2BP algorithm run on a 
   floorplan $f$ if there exists \textbf{block}\footnote{See \ref{Defn:Block} for the
   definition of a block} then there is an enveloping rectangle containing the 
   rooms labeled by the numbers in the block and nothing else, in $f$.
  \end{obs}
  \begin{proof}
    Let $\pi$ be the Baxter permutation produced by algorithm FP2BP
    when run on the mosaic floorplan $f$. Suppose there is a
    \textbf{block} at consecutive positions $i,\dots,j$ in $\pi$. If
    the block is a trivial \textbf{block}, then the observation is
    correct as there will be either just one number in the block or
    all the numbers from $1\dots n$ and in both cases rectangles
    labeled by the numbers in the block are contained inside trivial
    enveloping rectangles.  The remaining case is that the block is a
    non-trivial block.  That is there is at least one number in $[n]$
    which is not contained in the block.  Since the basic blocks in a
    mosaic floorplan are rectangular in shape, if the rooms which are
    labeled by the numbers in the block do not form an enveloping
    rectangle it must be forming a shape with at least one \textbf{T}
    shaped corner or they form disconnected clusters. If the rectangles
    form disconnected clusters and if there is at at least one cluster
    with a \textbf{T} shaped corner then this reduces to the case that
    the shape formed by the basic rectangle has one \textbf{T} shaped
    corner. Hence all of them must be forming clusters which are
    rectangular in shape. Take any two such disconnected clusters and 
    take all the basic rectangles between them, it is obvious that 
    after labeling the top cluster the basic rectangles between two
    clusters will be labeled before reaching the second cluster since
    it is not connected to the first. Hence it contradicts our assumption
    that the basic rectangles in consideration where labeled by elements
    in a \textbf{block} of a permutation as they do not form a range together.
    Hence it remains to prove that if there is \textbf{T} shaped corner
    in the shape formed by the basic rectangles labeled by the numbers 
    in the block, it also leads to a contradiction.
    Since
    there are no empty rooms in a mosaic floorplan and the block is a
    non-trivial block there should be at least one basic rectangle
    adjacent to this \textbf{T} shaped corner which is labeled with a
    number not contained in the block. Let us consider case 1 in
    Figure~\ref{fig:TShapedCorners} where basic rectangles `a' and `b'
    are part of the block in the permutation $\pi$ whereas `c' is
    not. In this case it is clear that among these three the algorithm
    will label `a' first, `c' second and label `b' the last. Hence it
    contradicts our assumption that there exists a block in $\pi$
    containing labels of `a' and `b' but not `c' as the label
    corresponding to `c' will be a number between the labels of `a' and
    `b'. Hence this case is not possible. Let us consider case 2 in
    Figure~\ref{fig:TShapedCorners}, again `a' and `b' are part of the
    assumed block in $\pi$ whereas `c' is not.  Here the order in
    which the basic rectangles `a',`b',`c' will be deleted is: `b'
    first, `c' the second and `a' the last. Hence it contradicts our
    assumption that there is a block in $\pi$ containing `a' and `b'
    but not `c' as in $\pi$ label of `c' will appear in between labels
    of `a' and `b'. Similarly it can be proved that any such
    \textbf{T}-corner configuration will result in a contradiction to
    our assumption that there is a block in $\pi$, such that the rooms
    labeled by the numbers in that block is not contained inside an
    enveloping rectangle in the corresponding mosaic floorplan. Hence the
    observation.
    \begin{figure}
      \centering
      
      \caption{T-shaped corners}
      \label{fig:TShapedCorners}

      \includegraphics[scale=0.4]{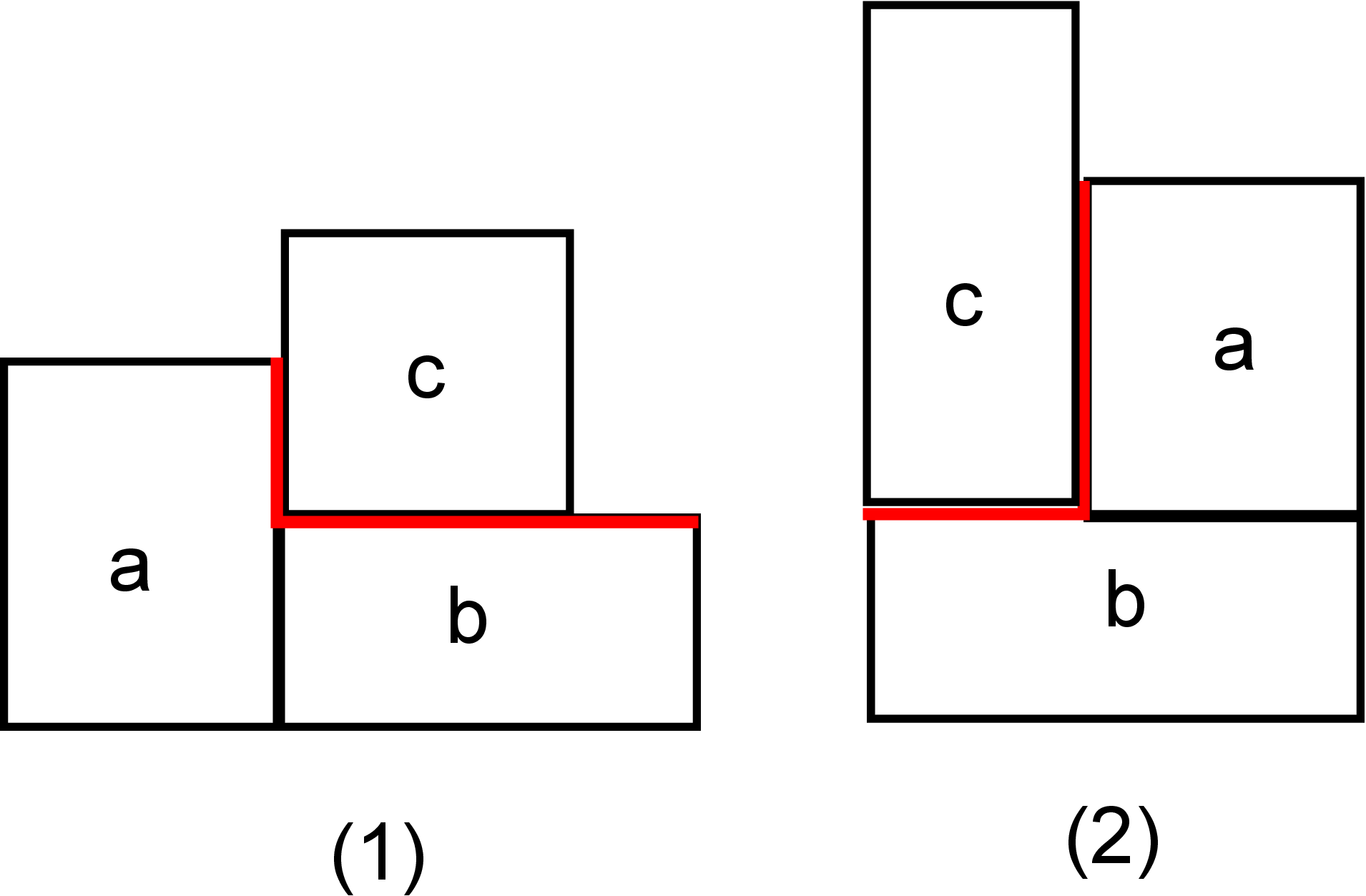}
    \end{figure}
  \end{proof}
  Now we will prove the characterization of Uniquely $\HFO_k$ 
  floorplans based on the permutations corresponding to them.

\begin{thm}
Uniquely $\HFO_k$ floorplans are in bijective correspondence with
permutations of length $k$ which are both Baxter and Simple.
\end{thm}
\begin{proof}

  The bijection is the bijection described by \cite{Ackerman20061674}
  from mosaic floorplans to Baxter permutations, restricted to
  Uniquely $\HFO_k$ floorplans. Since Uniquely $\HFO_k$ permutations
  are a subclass of $\HFO_k$ permutations which are in-turn a subclass
  of mosaic floorplans we know that Uniquely $\HFO_k$ floorplans
  correspond to a sub-family of Baxter permutations. So it remains to
  prove that they are also a sub-family of simple permutations of
  length $k$. Suppose $\pi$ is the Abe-label of a Uniquely $\HFO_k$
  floorplan which is not a simple permutation, then there exists a
  non-trivial block in $\pi$ consisting of $j,1<j<k$ numbers. By
  observation \ref{Obs:BlockEnvolopingRectangle} there is an
  enveloping rectangle containing just the rooms which are labeled by
  the numbers in the non-trivial block. Now we can obtain the $\HFO_k$
  floorplan corresponding to $\pi$, by removing the rooms labeled 
  by numbers in the non-trivial block and then placing the mosaic
  floorplan constituted by the rooms labeled by the numbers in the
  non-trivial block of $\pi$. Thus the floorplan is $\HFO_{\max \lbrace k-j,j \rbrace }$
  contradicting our assumption that it is Uniquely $\HFO_k$. Hence
  the Abe-label corresponding to a Uniquely $\HFO_k$ permutation
  has to be a simple permutation of length $k$.
\end{proof}

\section{Generating trees of Order $k$}

A generating tree for a mosaic floorplan is a rooted tree which represents
how the basic rectangle was embedded with successive mosaic floorplans
to obtain the final floorplan. A generating tree is called a generating
tree of order $k$ if it satisfies the following properties:
\begin{itemize}
\item All internal nodes are of degree at most $ $$k$.
\item Each internal node is labeled by a Uniquely $\HFO_l$ permutation($l\leq k)$,
representing the mosaic floorplan which was embedded.
\item Out degree of a node whose label is a permutation of length $l$ is
$l$.
\item Each leaf node represents a basic room in the final floorplan and
is labeled by the Abe-label of that room in the floorplan.
\end{itemize}
The internal nodes are labeled by permutations corresponding to
Uniquely $\HFO_l$ floorplans because they are the only $\HFO_l$
floorplans which cannot be constructed hierarchically with $\HFO_j$
floorplans for $j<l$.  By this definition there is at least one
generating tree of order $k$ for any $\HFO_k$ floorplan. But the
problem is that due to the symmetry associated with vertical and
horizontal cut operations there could be multiple generating trees
representing the same floorplan.  To avoid this problem we define
skewed generating trees. An order $k$ generating tree is called a
\textbf{skewed generating tree} of order $k$ if it satisfies
additional to the above rules the following rule:
\begin{itemize}
\item The right child of a node cannot be labeled the same as parent if
the parent is labeled from $\left\{ 12,21\right\} $.\end{itemize}
\begin{thm}
$\HFO_k$ floorplans with $n$ rooms are in bijective correspondence
with skewed generating trees of order $k$ with $n$ leaves.
\end{thm}
Clearly the additional rule introduced above removes the symmetry
associated with vertical(permutation $21$) and horizontal(permutation
$12$) cut operations as we have seen in Slicing trees. Hence it remains
to prove that for any other embedding such a symmetry doesn't exist
thus making the skewed generating tree unique for an $\HFO_k$ floorplan.
Note that the generating tree provides a hierarchical decomposition
of the permutation corresponding to the floorplan into blocks as illustrated
by the figure \ref{Flo:GeneratingTree}. Albert and Atkinson \cite{Albert20051}proved
the following :

\begin{figure}
\caption{Generating Trees of Order $k$}
\label{Flo:GeneratingTree}
\begin{center}\includegraphics[scale=0.4]{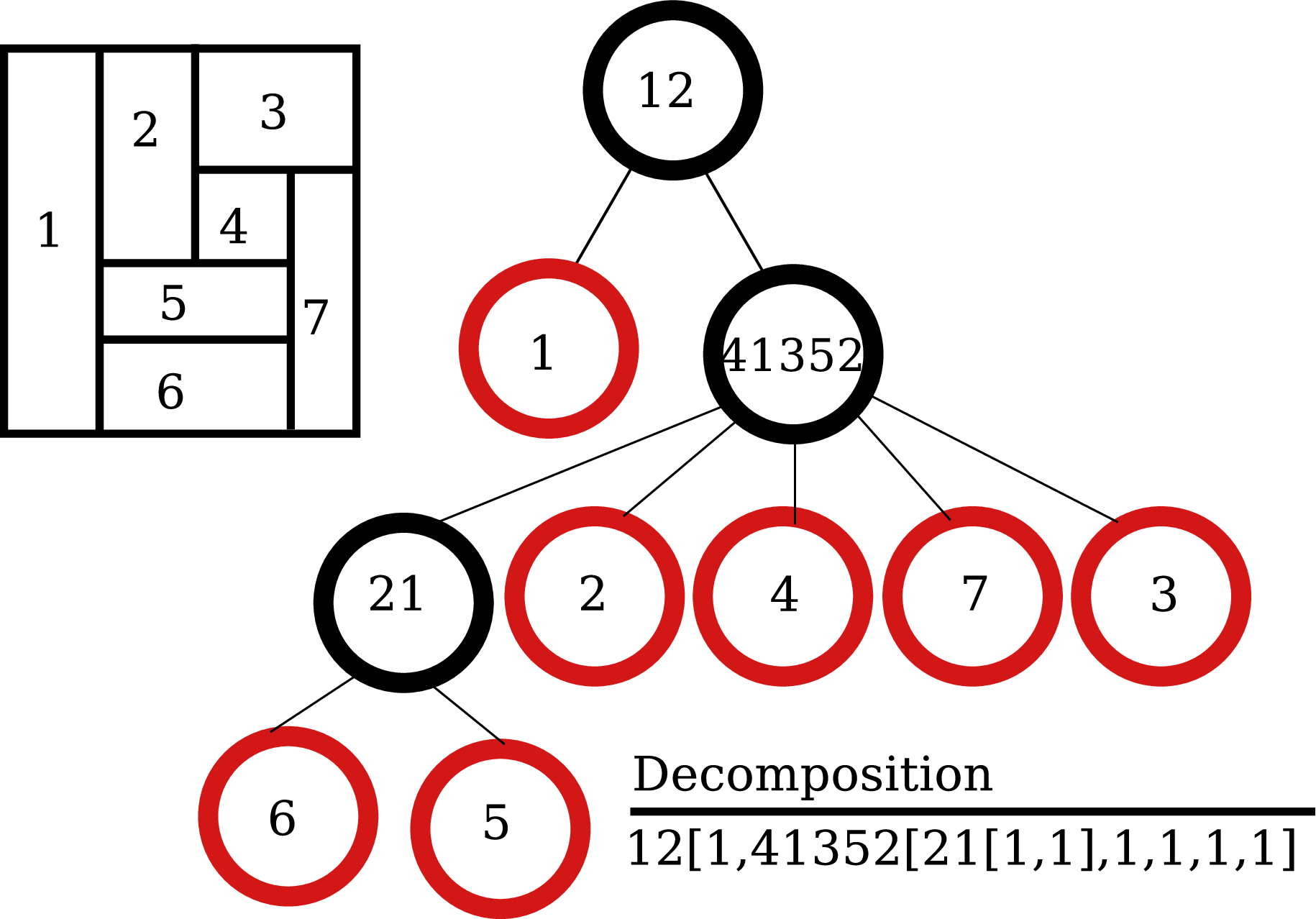}\end{center}
\end{figure}

\begin{thm}
[M.H Albert, M.D Atkinson]For every non-singleton permutation $\pi$
there exists a unique simple non-singleton permutation $\sigma$ and
permutations $\alpha_{1},\alpha_{2},\alpha_{3},\alpha_{4},\dots,\alpha_{n}$
such that \[
\pi=\sigma[\alpha_{1},\alpha_{2},\alpha_{3},\alpha_{4},\dots,\alpha_{n}]\]

Moreover if $\sigma\neq12,21$ then
$\alpha_{1},\alpha_{2},\alpha_{3},\alpha_{4},\dots,\alpha_{n}$ are
also uniquely determined. If $\sigma=12$(respectively $21$) then
$\alpha_{1}$ and $\alpha_{2}$ are also uniquely determined subject to
the additional condition that $\alpha_{1}$ cannot be written as
$(12)[\beta,\gamma]$(respectively as $(21)[\beta,\gamma]$.
\end{thm}
The proof is completed by noting that the decomposition obtained by
skewed generating tree of order $k$ satisfies the properties of their
decomposition, that is if $\sigma=12/21$ its right child cannot be
$12/21$ hence the block corresponding to the right child, $\alpha_{1}$
cannot be $(12)[\beta,\gamma]/(21)[\beta,\gamma]$. Since such a decomposition
is unique the skewed generating tree also must be unique. Hence the
theorem. This bijection between $\HFO_k$ floorplans is very crucial 
for the characterization of $\HFO_k$ floorplans in terms of permutations
corresponding to it and also for getting a coded representation of 
$\HFO_k$ floorplans for stochastic search methods. 

\section{Characterization of $\HFO_k$}
\begin{thm}
\label{thm:HFO-k-bijection}$\HFO_k$ floorplans with $n$ rooms is in bijective
correspondence with Baxter permutations of length $n$  which avoids patterns
from Simple permutations of length $k+1$ and Exceptionally simple permutations
of length $k+2$.
\end{thm}

To prove this we need the following theorem by 
\begin{proof}The bijection is the bijection defined by \cite{Ackerman20061674}
 from mosaic floorplans to Baxter permutations
restricted to $\HFO_k$ floorplans. Since $\HFO_k$ is sub-class of
mosaic floorplans the bijection will map them to a sub-class of Baxter
permutations. It is easy to prove that if a permutation corresponds
to an $\HFO_k$ floorplan then it cannot contain text which matches
patterns from Simple permutations of length $k+1$ and exceptionally
simple permutation of length $k+2$. Suppose in the permutation $\pi$
corresponding to an $\HFO_k$ floorplan 
there is text at comprising of points $(i_{1},i_{2},i_{3},i_{4},\ldots,i_{j})$
which matches a simple permutation $\sigma$ of length $j,j>k$.
Then in the generating-tree of order $k$ of the $\HFO_k$ floorplan
corresponding to the given permutation $\pi$, it is clear that no
proper subset of $\left\{ \pi[i_{m}]|1\leq m\leq j\right\} $ could
be inside a single sub-tree because in the generating tree the elements
of the sub-tree will always form a range(root of node of 
the subtree corresponds to the enveloping rectangle of all leaf nodes
in the subtree) and no proper subset of a
simple permutation can form a range. Consider the smallest(in the number 
of vertices) subtree which
contains all of $\left\{ \pi[i_{m}]|1\leq m\leq j\right\} $. In this subtree
let the root node be $r$ and let its children be $\lbrace r_1,r_2,r_3,\dots,r_l \rbrace$
. None of the subtrees rooted at $r_i,1\leq i \leq l$ can contain all of
$\left\{ \pi[i_{m}]|1\leq m\leq j\right\} $ because then $r_i$ will be the 
smallest subtree containing all of $\left\{ \pi[i_{m}]|1\leq m\leq j\right\} $.
And for the above mentioned reason no $r_i, 1\leq i \leq l$ can contain
a proper subset of $\left\{ \pi[i_{m}]|1\leq m\leq j\right\} $. Hence 
there should be $j$ children of $r$, each containing exactly one node from
$\left\{ \pi[i_{m}]|1\leq m\leq j\right\} $.
Hence there are at least $j$ children for the root.
Since $j>k$ this leads to a contradiction to our assumption
that the permutation corresponds to an $\HFO_k$ floorplan because there can no internal
node of degree strictly greater than $k$ in a generating tree of order $k$.
. 
So it remains to prove that any $\HFO_l$, $l>k$ floorplan which is
not $\HFO_k$ will contain a text matching the patterns from either
Simple permutations of length $k+1$ or $k+2$. Let the floorplan
be $\HFO_l$ for $l>k$ and which is not $\HFO_k$, and $l$ be the smallest
such integer that the floorplan is $\HFO_l$. That is in the floorplan
tree for this floorplan there is an internal node with out-degree $l$.
This node will correspond to a Uniquely $\HFO_l$ permutation 
and the ranges formed by subtrees rooted at the children of this
node will be form the pattern which is the Uniquely $\HFO_l$ permutation
corresponding to the root node. Hence to obtain the \textbf{text} matching
the \textbf{pattern} in the
permutation corresponding to the floorplan we can pick
one arbitrary leaf node from each subtree and then choose the Abe-label 
of that node.
Hence every simple permutation of length $l$ contains a pattern from
either simple permutations of length $l-1$ when the original permutation
is not exceptionally simple or simple permutations of length $l-2$
when the original permutation is exceptionally simple as deletion
of an element from a permutation preserves the relative ordering among
the other elements of the permutation. So we can find in an $\HFO_l$,
$l>k$ permutation a pattern which is a simple permutation of length
$k+1$ or $k+2$ by applying the above observation recursively. \end{proof}

\section{Algorithm for Recognition}

The algorithm is based on the bijection we obtained above. If a given
permutation is Baxter then it is $\HFO_{j}$ for some $j$. Suppose
it is $\HFO_{k}$ then we know that there exits an order $k$ generating
tree corresponding to the permutation. And in a generating tree of
order $k$ the label of leaves of any sub-tree will always form a
range as the root of the sub-tree is an enveloping rectangle which
contains all the rooms corresponding to the leaves. The algorithm~\ref{alg:algor-recogn}
tries to iteratively reduce the sub-trees of the generating tree to
nodes, level by level.

We will prove the correctness of the algorithm by use of the following loop 
invariant.

\textbf{Loop Invariant:}
  At the end of each iteration of the \textbf{for loop} of lines 2-13,
  all sub-trees of the generating tree containing leaf nodes which are
  labeled only from $\left\{ \pi[j]|1\leq j\leq i\right\} $ are replaced
  with a single node(correspondingly pushed onto the stack as a range
  of numbers which are labels of the leaf nodes of the sub-tree).

\textbf{Initialization:}
When $i=1$, $\left\{ \pi[j]|1\leq j\leq i\right\} $ is equal to $\pi[1]$.
Since
all internal nodes are of out-degree $2$ or more the only sub-tree
containing only $\pi[1]$ is the leaf node itself so there is nothing to
be reduced hence the condition is trivially met.

\textbf{Maintenance:}
We will assume that all the sub-trees whose leaves
are labeled from $\left\{ \pi[j]|1\leq j\leq i\right\} $ is reduced
to a node before iteration $i+1$ and then prove that at iteration $i+1$ the condition is maintained
 by the
\textbf{for loop} for all sub-trees whose leaves are labeled from $\left\{ \pi[j]|1\leq j\leq i+1\right\} $.
Suppose if there is a sub-tree whose leaves are labeled only from $\left\{ \pi[j]|1\leq j\leq i+1\right\} $
and does not contain $\pi[i+1]$ then by the induction hypothesis
it has been reduced to a node. Suppose there exists sub-trees which
also contains $\pi[i+1]$ as a leaf node then choose the sub-tree
which has $\pi[i+1]$ as an immediate child node. In this sub-tree
all its other children are reduced to nodes by induction hypothesis,
so at iteration $i+1$ there must exist $j\leq k$ elements at the top of the
stack corresponding to the children of this sub-tree as it has
$\pi[i+1]$ as the right most leaf node which also is the current stack
top. Now the algorithm will reduce
those $j$ elements to a range and then try to reduce the tree further
by scanning the top $k$ elements of the stack.

\textbf{Termination:}
When $i=n$ the tree itself is a sub-tree containing leaf nodes
labeled from $[n]$ hence it must be reduced to a single node. Hence
if at the end of the algorithm the stack contains just one element
that would mean that the given permutation is $\HFO_{k}$. Suppose
the algorithm is able to reduce it to a single element on the stack
then by retracing the stack operations carried out by the algorithm 
we can build an order $k$ generating tree as at any point of time 
we merged at most $k$ elements together which together formed a range and
was a Baxter permutation(thus correspond to a mosaic floorplan).
Hence upon acceptance by the algorithm for a given permutation
it is clear that there is an order $k$ generating tree corresponding
to the given permutation. If the permutation is not
$\HFO_{k}$ algorithm would not be able to find a generating tree
of order $k$. Hence it would reject such a permutation.

 Figure \ref{fig:RecognitionAlgoExample}
illustrates the working of the algorithm on an $\HFO_5$ permutation.
The figure shows the generating tree of order $5$ corresponding
to the floorplan, and trace of the stack used by the algorithm
(to be read from left).
The permutation is scanned from left to right and each time an
insertion takes place in the stack, the top $5$ elements are searched
to see if they form a range. In the example shown in the figure until
$3$ is inserted onto the stack this doesn't happen. At the instant $3$
is inserted it is combined with the other four elements to a range 
corresponding to the internal node labeled $41352$ in the generating
tree. Then this is combined with $1$ to form another range and finally
it is reduced to a single node by combining with $7$. This final node
corresponds to the root node of the generating tree.
\begin{algorithm}
 \SetKwData{Stck}{S}
 \SetKwData{Rng}{R}
 \SetKwData{Top}{top}
 \SetKwData{I}{i}
 \SetKwData{J}{j}
 \SetKwData{K}{k}
 \SetKwData{Lvar}{l}
 \SetKwData{N}{n}

 \SetKw{KwDownTo}{downto}
 \SetKw{KwAccept}{Accept}
 \SetKw{KwReject}{Reject}

 \SetKwFunction{Push}{push}
 \SetKwFunction{Pop}{pop}
 \SetKwFunction{Range}{Range}
 \SetKwFunction{Size}{size}

 \SetKwInOut{Input}{Input}

 \Input{A permutation $\pi$ of length $n$}
 \BlankLine
 Stack \Stck$\leftarrow \phi$\;

 \For{\I$=1$ \KwTo \N} {
    \Stck.\Push{\Range{$\pi[i]$}}\;
    \While{There exists a \J, \J$\leq$\K such that \J is the least such number
for  which top \J elements of \Stck form a range} {
     \If{\Stck$[\Top \dots(\Top-\J)]$ is a Baxter permutation}{
          \Rng = \Range{\Stck$[$\Top$\dots($\Top-\J$)]$}\;
          \For{\Lvar$=$\J \KwDownTo $1$ } {
             \Stck.\Pop{}\;
          }
          \Stck.\Push{\Rng}\;
     }
    }
 }
 \eIf{\Stck.\Size{}$=1$} {
 \KwAccept\;
 } {
  \KwReject\;
 }

 \caption{Algorithm for checking if a permutation is $\HFO_k$}
 \label{alg:algor-recogn}
\end{algorithm}

\begin{figure}
  \centering
  
  \caption{Example : $\HFO_5$ recognition algorithm}
  \label{fig:RecognitionAlgoExample}
  \includegraphics[scale=0.4]{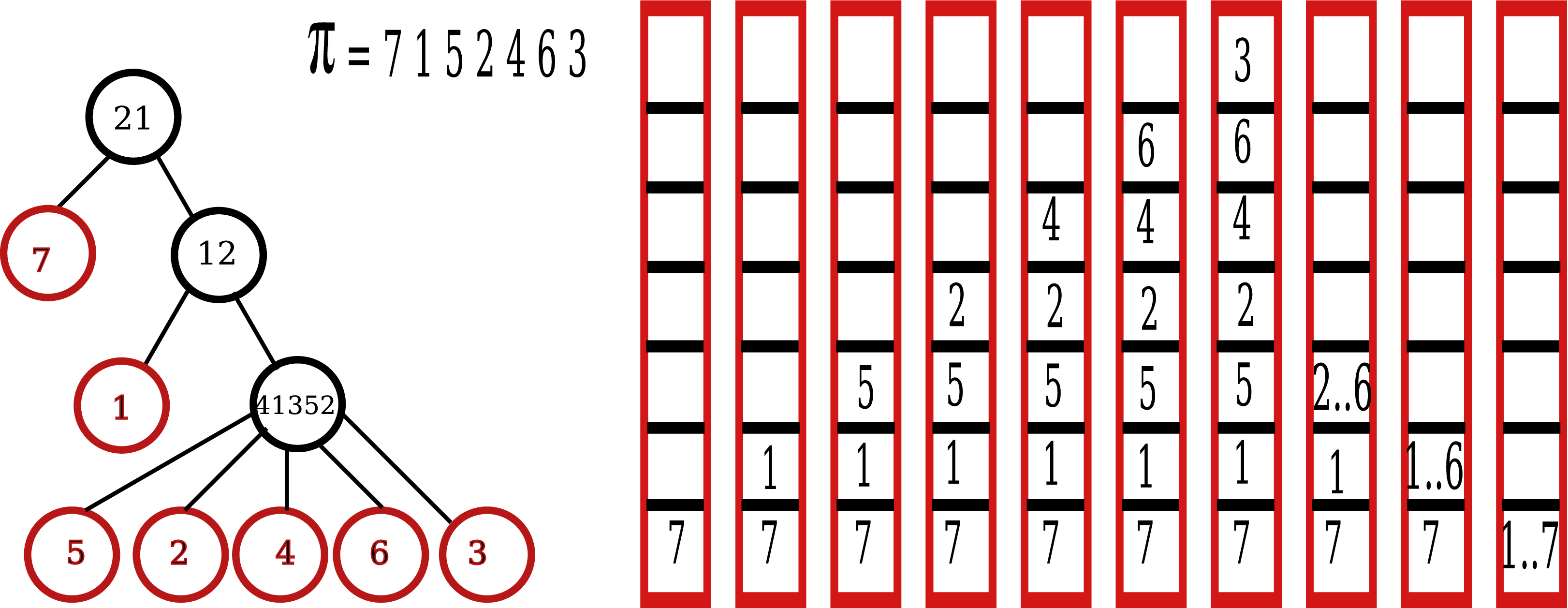}
\end{figure}

\section{Analysis of the recognition algorithm}

The algorithm runs in both linear time and space for a fixed predetermined
value of $k$ which does not change with the input length. Linear space is easy to
observe as the stack at any point of the execution of the algorithm contains
no more than $n$ elements. To prove that the algorithm runs in linear time
we assign an amortized cost of $k^2$-units to each node(including leaf nodes)
in the generating tree.
We also use the observation that in a tree of $n$ nodes there can not
be more than $n-1$ internal nodes. Hence the total nodes of the generating tree is
bounded by $2n-1$. So if the algorithm spends at most $k^2$ units of time with
each node then the total time taken by the algorithm is $O(n)$.

Now we will 
prove that the algorithm spends at most $k^2$ units of time with each node 
in the skewed generating tree of order $k$ corresponding to the permutation, if the
permutation is $\HFO_k$. 
The key operation in the algorithm is scanning the top $k$ elements of the stack
to find a set of elements which form a range. It is easy to observe that the stack
is scanned only when a new element is inserted onto the top of the stack. The newly
inserted element can either be a number in the permutation(which corresponds to a leaf node
in the order $k$ generating tree) or a range of elements(which corresponds to an internal node
in the order $k$ generating tree). Also, observe that a node is inserted only once into the stack.
And when a range corresponding to a node is inserted to the stack,
it is either merged with the top $j, j<k$ elements of the stack to become another
node or the top $k$ elements of the stack are searched unsuccessfully and the node remains on the 
top of the stack. In both cases, at most $k^2$ units of time is spend. 
Because to check whether top $i$ elements form a range, $i$ units of time is 
needed, so doing this for all $1\leq i \leq k$ we need $\frac{k(k+1)}{2}$ time
which is clearly upper bounded by $k^2$.
 Thus distributing the costs this way, we get
that for each node in the tree at most $k^2$-units of time is spend. Since there are
only $O(n)$ nodes in the tree the total time spend by the algorithm is 
$O(n)$.

If the permutation is not Baxter then at some point during the  execution of the algorithm
it will find a set of ranges on stack top which does not form a Baxter permutation,
or the algorithm would not be able to merge the elements of the permutation to a single
node. Even in this case the number of nodes in the partial tree which the 
algorithm can find with the given permutation is bounded by $2n$, and with
each node at most $k^2$ units of computation will be spend. Hence in this
case also the algorithm runs in linear time. If the permutation is Baxter 
and is $\HFO_j$ for some $j>k$ and is not $\HFO_k$ then again the same cost
analysis is  valid  and hence the algorithm runs in linear time for all possible
types of input permutations.

Note that checking if a set of $k$ elements form a range can be checked
in constant time for a fixed value of $k$ by writing conditional statements to check
if the elements follow any of the $k!$ arrangements. We can also check if a 
set of $k$ elements form a Baxter permutation for a fixed $k$ in constant time 
by writing conditional statements to check if their rank ordering is equivalent
to any one of the Baxter permutations of length $k$(whose number is bounded by number
of permutations, $k!$). Hence the above algorithm runs in linear time for a predetermined
value of $k$.

If the value $k$ is unknown the same algorithm can be made to run in $O(n^2 \log_2 n)$
time to find out the minimum $k$ for which the input permutation is $\HFO_k$ with
 some simple modifications in the implementation.
 The first modification we have to implement is to make the algorithm
 checks if the input permutation $\pi$ is Baxter permutation. If it is not 
it cannot be $\HFO_j$ for any $j$ hence is rejected. If it is a Baxter permutation
then we know that it is $\HFO_k$ for some $ k \leq n $. And also at each time a 
new element is inserted onto the stack we have to check if that forms a range with any 
of the top $j, j \leq |S|$ elements of the stack where $|S|$ denotes the current 
size of the stack. Implementing these changes alone we obtain  the modified 
Algorithm \ref{Flo:AlgoHFOkFindingK}. The
increase in running time comes from the fact that we don't know the value of the
$k$, thus forcing us to scan the entire stack at the insertion of a new element 
on top of the stack costing us $cn \log_2 n$ time to sort the elements of the 
stack and see if there exists a $j, j \leq |S|$  such that the current element
forms a range along with $S[top,\dots,(top-j)]$. Checking if a permutation is
Baxter takes $O(n^2)$ time. And we use the same amortized cost analysis as above
but with each node(internal or leaf) in the tree we associate the cost of $cn \log_2 n$
which is spent at the time it is first inserted on to the stack for sorting
the current elements of the stack. The number of nodes in the tree is again bounded by $2n$. Hence the
stack reduction part of the algorithm runs in $O(n^2 \log_2 n)$ time and checking if a
permutation is Baxter part runs in time $O(n^2)$. So the
total time taken is $O(n^2 \log_2 n)$. 

\begin{algorithm}
 \SetKwData{Stck}{S}
 \SetKwData{Rng}{R}
 \SetKwData{Top}{top}
 \SetKwData{I}{i}
 \SetKwData{J}{j}
 \SetKwData{K}{k}
 \SetKwData{Mvar}{m}
 \SetKwData{Lvar}{l}
 \SetKwData{N}{n}
 \SetKwData{Array}{Array}
 \SetKwData{Mnodes}{mergedNodes}

 \SetKw{KwAccept}{Accept}
 \SetKw{KwReject}{Reject}
 \SetKw{KwTrue}{true}
 \SetKw{KwFalse}{false}
 \SetKw{KwBreak}{Break}
 \SetKw{KwDownTo}{downto}

 \SetKwFunction{Push}{push}
 \SetKwFunction{Pop}{pop}
 \SetKwFunction{Range}{Range}
 \SetKwFunction{Size}{size}
 \SetKwFunction{Sort}{sort}
 \SetKwFunction{Start}{start}
 \SetKwFunction{End}{end}

 \SetKwInOut{Input}{Input}

 \SetKwComment{tcc}{/*}{*/}

 \Input{A permutation $\pi$ of length $n$}
 \BlankLine 
 Stack \Stck$\leftarrow \phi$\;

 \If{$\pi$ is \textbf{not} a Baxter permutation}{
   \KwReject\;
 }

 \For{\I$=1$ \KwTo \N} {
   \Stck.\Push{\Range$(\pi[i])$}\;
   \Mnodes$\leftarrow$\KwTrue\;
   
    \Repeat{\Mnodes$=$\KwFalse}{
      \Array$\leftarrow$ \Sort{\Stck$[$\Top$\dots,1)]$} \;
     \tcc{Find the longest range containing stack top}
      %              by searching to the left and right of stack 
       %             top in the sorted array till you cannot extend
        %            the range any further}
      \For{\I$=1$ \KwTo \N}{
        \If{\Array$[i]=$\Stck$[$\Top$]$}{
          \KwBreak\;
        }
      }
      \For{\J$=$\I \KwTo \N} {
        \tcc{This loop will run until \Array$[$\I$,\dots,$\J$]$ cease to become a range of
                       contiguous elements}
        \If{\Array$[\J]$.\End{}$+1\neq$\Array$[\J+1]$.\Start{}}{
           \KwBreak\;
        }
      }
      \For{\Lvar$=$\I \KwDownTo $2$} {
        \tcc{This loop will run until \Array$[$\Lvar$,\dots,$\I$]$ cease to become a range of
                       contiguous elements}
        \If{\Array$[\Lvar]$.\Start{}$-1\neq$\Array$[\Lvar-1]$.\End{}}{
           \KwBreak\;
        }
      }
      \tcc{\Array$[\Lvar,\dots,\J]$ forms a range containing stack top, hence \Stck$[\Top,\dots,(\Top-(\J - \Lvar))]$ 
           forms a range}
     \eIf{\J$\neq$\Lvar}{
        \Rng$\leftarrow$\Range{\Stck$[\Top,\dots,\Top-(\J-\Lvar)]$} \;
        \For{\Mvar$=1$ \KwTo $(\J-\Lvar)$}{
          \Stck.\Pop{}\;
        }
        \Stck.\Push{\Rng}\;
     }{
       \Mnodes$\leftarrow$\KwFalse\;
     }

    }
  }
  \eIf{\Stck.\Size{}$=1$}{
    \KwAccept \;
  }{
    \KwReject \;
  }

 \caption{Algorithm for finding the minimum $k$ for which $\pi$ is $\HFO_k$}
 \label{Flo:AlgoHFOkFindingK}
\end{algorithm}

\section{Counting}

Given an $n$, it is interesting to know the number of distinct $\HFO_k$
floorplans with $n$ rooms. We call two $\HFO_k$ floorplans distinct
in the same way \cite{Sakanushi1193019} defines it.
Given a floorplan $f$, a segment $s$ supports a room $r$ in $f$
if $s$ contains one of the edges of $r$. We say that $s$ and $r$
hold a \textit{top-,left-,right-, or bottom-seg-room} relation if
$s$ supports $r$ from the respective direction. Two floorplans are
equivalent if there is a labeling of their rectangles and segments
such that they hold the same seg-room relations, otherwise they are
distinct. This is the same definition of equivalent floorplans 
\cite{Ackerman20061674} used. Since we are considering a restriction of the
bijection they gave between mosaic floorplans and Baxter permutations
to $\HFO_k$ floorplans, we can say that two $\HFO_k$ floorplans
are distinct if they are mapped to different permutations by this
bijection. And by theorem-\ref{thm:HFO-k-bijection} there is a bijection
between $\HFO_k$ floorplans and Baxter permutations which avoid patterns
from simple permutations of length $k+1$ and exceptionally simple
permutations of length $k+2$. Hence we give a lower bound 
on number of distinct $\HFO_k$ floorplans on $n$ rooms
by giving a lower bound(resp., an upper bound) on the number of $\HFO_k$
permutations.
\begin{thm}
There are at least $3^{n-k}$ $\HFO_k$ permutations of length $n$
which are not $\HFO_j$ for $j<k$. 
\end{thm}
\begin{proof}The proof is inspired by the insertion vector scheme
  introduced by \cite{Chung1978382} to enumerate the admissible
  arrangements for Baxter permutations. The idea is to start with a
  Uniquely $\HFO_k$ permutation which is of length $k$ say $\pi_{k}$
  and successively insert $(k+1,k+2,k+3,k+4,\ldots,n)$ onto it such a
  way that we are guaranteed that it remains both Baxter and no
  patterns from simple $k+1$ or exceptionally simple $k+2$ is
  introduced so the final permutation $\pi_{n}$ is of the desired
  property. It is very clear that if you insert $k+1$ onto two
  different positions of $\pi_{k}$ you get two different
  permutations. It is also not hard to see that if you start with two
  different permutations $\pi_{i}^{'}$ and $\pi_{i}^{''}$ then there
  is no sequence of indices to which insertion of
  $(i+1,i+2,i+3,i+4,\ldots,n)$ will make the resulting permutations
  the same. Hence by counting the number of ways to insert
  $(k+1,k+2,k+3,k+4,\ldots,n)$ successively, a lower bound on the
  number of $\HFO_k$ permutations is obtained.  So the problem boils
  down to counting the number of ways to insert $i+1$ given a
  permutation $\pi_{i}$ which is $\HFO_k$ but not $\HFO_j$ for
  $j<k$. We do not have an exact count for this but it is easy to
  observe that in such a permutation $\pi_{i}$ there are always four
  locations which are safe for insertion of $i+1$ irrespective of
  relative order of elements of $\pi_{i}$. 
  By safe we mean that insertion of $i+1$ to $\pi_i$ would 
  neither make the resulting permutation non-Baxter nor will it make
  a non-simple permutation. The four safe locations are : 
\begin{enumerate}
\item Before the first element of $\pi_{i}$.
\item After the last element of $\pi_{i}$.
\item Before $i$ in $\pi_{i}$.
\item After $i$ in $\pi_{i}$.
\end{enumerate}
Let us prove that these sites are actually safe for insertion of $i+1$.
We will first prove that insertion of $i+1$ onto these sites cannot
introduce a pattern which matches a simple permutation of length $j,j>k$.
Suppose $i+1$ is inserted before or after $i$ in $\pi_{i}$ and
the newly obtained permutation has a \textbf{text} which matches a simple
permutation of length $j,j>k$ .The \textbf{text}  must be involving $i+1$ as
otherwise $\pi_{i}$ will also contain the same pattern. The text
matching the pattern cannot involve $i$ also, as if it does then the
pattern will have two consecutive integers corresponding to the location
of $i$ and $i+1$ in the text making it not a simple permutation
. Thus the text matching the pattern must involve $i+1$ and it must
not involve $i$, but then replacing $i$ by $i+1$ we get a text
in $\pi_{i}$ matching the same pattern contradicting our assumption
that $\pi_{i}$ is $\HFO_k$. Now it remains to prove that inserting
$i+1$ before or after $\pi_{i}$ is safe. 
Suppose insertion
of $i+1$ before or after $\pi$ introduces a text matching a simple
permutation of length $j>k$, then the text must involve $i+1$. But
since $i+1$ is greater than any other element in $\pi_{i}$ in the
pattern of length $j$, $i+1$ will be matched with the number $j$. But then it
would mean that pattern is a permutation $\sigma$ on $[j]$ which
has $j$ as its first/last element as $i+1$ is placed after or before $\pi_i$.
This contradicts our assumption that
the pattern is a simple permutation as $\sigma$ maps either $\left\{ 2,\ldots,j\right\} $
to $\left\{ 1,\ldots,j-1\right\} $ or $\lbrace 1,\dots,j-1 \rbrace$ to $ \lbrace 1,\dots,j-1 \rbrace$ which is a proper sub-range. Hence it is not
possible that insertion of $i+1$ onto any of these locations introduces
a text matching a pattern from simple permutations of length $j,j>k$.
Now it remains to prove that the insertion of $i+1$ cannot
introduce any text which matches $3142/2413$ with absolute difference
between first and last being one.
Suppose it did, then it has to involve $i+1$
since $\pi_{i}$ is Baxter, and if it involves $i+1$, $i+1$ will
have to match $4$ in $3142/2413$ as there is no element greater
than $(i+1)$ in $\pi_{i}$. But $i+1$ matching $4$ is not possible
because in the first case there is nothing to the left of $i+1$,
in the second case there is nothing to the right of $i+1$, and in
third and fourth cases this is not possible for the reason that if
$2413/3142$ involves both $i$ and $i+1$ then $i$ has to match
$3$ and $i+1$ has to match $4$ as they are the second largest and
largest elements in the new permutation but this is not possible in
these cases as $i$ is adjacent to $i+1$ and there cannot be any
element matching $1$ in between them. Hence in these cases the only
possibility left is that $i+1$ is matched to $4$ in $3142/2413$
but the text matching the pattern does not involve $i$ and since
$i$ is adjacent to $i+1$ and greater than any element of $\pi_{i}$
it can be replaced for $i+1$ to get $3142/2413$ in $\pi_{i}$ with
the absolute difference between first and last being one, contradicting
the fact that $\pi_{i}$ is Baxter. Hence we have proved that introduction
of $i+1$ in these sites are safe.

Note that even though we have identified four safe locations for insertion
of $i+1$ into a $\pi_{i}$ sometimes $i$ could be the first element
of the permutation $\pi_{i}$ thus making the location before $i$
and location before $\pi_{i}$ one and the same. Similarly if $i$
is the last element the location after $i$ and location after $\pi_{i}$
also coincides. But for any permutation $\pi_{i}$ only one of the
above two conditions can occur, so there are always three distinct locations
to insert $i+1$. Now by starting from a Uniquely $\HFO_k$ permutation
we can get $3^{n-k}$ different permutations by inserting successive
elements from $\left\{ k+1,k+2,k+3,k+4,\ldots,n\right\} $. Hence
the theorem.

\end{proof}

%%%%%%%%%%%%%%%%%%%%%%%%%%%%%%%%%%%%%%%%%%%%%%%%%%%%%%%%%%%%
% Simulated Annealing Moves for $\HFO_k$ floorplans
\section{Simulated Annealing for $\HFO_k$ family of floorplans}
Wong and Liu\cite{wong1986new}
designed a set moves for $\HFO_2$ floorplans based on the post-order traversal
of the the corresponding skewed order $2$ generating tree. Later \cite{WongAndThe}
extended this idea to $\HFO_5$ floorplans. Now based on our result which unified
the way $\HFO_k$ floorplans are represented using generating trees we can easily
extend the moves defined by \cite{WongAndThe} to $\HFO_k$ floorplans. We also prove
that the solution space thus obtained is connected and is of diameter $O(n^2)$.
We also prove that our solution space is \textbf{P-admissible} except for the
last property which requires the search space to include the optimal floorplan 
for a given floorplanning problem. Almost all of the solution spaces for floorplanning
problem cannot guarantee this property. This is because the optimal solution to
floorplanning problem may contain empty rooms and finding the number of optimal
empty rooms for an instance of floorplanning problem is in itself an open problem.
\subsection{Simulated Annealing Moves for $\HFO_k$ floorplans}

The moves described in \cite{WongAndThe} can be easily generalized to any $\HFO_k$ 
provided that you can capture the floorplan using a floorplan tree
, find out Uniquely $\HFO_l$ floorplans for $l \leq k$ - the
 internal nodes of the tree -, and a nice representation
of these floorplans to serve as operators in the normalized polish
expression. We have already proved that $\HFO_k$ floorplans are 
in bijective correspondence with skewed generating trees of order $k$.
 We also provided an algorithm(Algorithm \ref{Flo:AlgoHFOkFindingK}) to
find the minimum value of $k$ for which a given Baxter permutation is
also a permutation corresponding to an $\HFO_k$ floorplan. Hence 
we can run this algorithm on all Baxter permutations of length $k$
and find out permutations corresponding to Uniquely $\HFO_k$ permutations
because Uniquely $\HFO_k$ are $\HFO_k$ floorplans such that $k$ is
the minimum such integer for which they are $\HFO_k$. The bijection
given by \cite{Ackerman20061674} can be used to represent a
Uniquely $\HFO_k$ floorplan as permutation of length $k$.
We assume an implicit left-to-right ordering among the children
of internal nodes in generating trees of order $k$ and then
 use the post-order traversal of the tree to represent corresponding floorplan.
To 
distinguish operators from operands we enclose permutation
corresponding to Uniquely $\HFO_k$ in set of $\left[ \right]$ parenthesis.

We will now formally define \textbf{normalized polish expressions
of length $k$} which corresponds to post order traversal of a skewed
generating tree of order $k$. A normalized polish expression of length $k$ 
is a sequence $\alpha = \alpha_1,\alpha_2,\alpha_3,\dots,\alpha_m$ of elements from
$\lbrace 1,2,3,\dots,n,\lbrace [\pi] | \pi \in S_j, j\leq k \rbrace \rbrace$ satisfying the following conditions.
Let $x_{ji}$ represent the number of operators which are 
permutations of length $j$ enclosed within $\left[ \right]$ brackets in 
the sequence $\alpha_1,\alpha_2,\alpha_3,\dots,\alpha_i$. And $y_i$ represents the number
of operands in the sequence $\alpha_1,\alpha_2,\alpha_3,\dots,\alpha_i$. 
\begin{itemize}
\item for each $j\leq n$ there exists a unique index $k$ such that $\alpha_k = j$.
\item $\sum_{j=2}^{k}(j-1)x_{ji}< y_i$, for all $i=1,2,3,\dots,m$.
\item $\alpha_i\alpha_{i+1} \neq [12][12]$ and $\alpha_i\alpha_{i+1} \neq [21][21]$ for each $i$ in $1,2,3,\dots,m-1$.
\end{itemize} 

Figure~\ref{fig:ExampleNormPolishExpOrderK} shows the normalized $2-5$ polish expression corresponding
to the $\HFO_7$ floorplan in the figure.

\begin{figure}
  \centering
  
  \caption{Example normalized polish expression of order $k$ and corresponding floorplan}
  \label{fig:ExampleNormPolishExpOrderK}
  
  \includegraphics[scale=0.35]{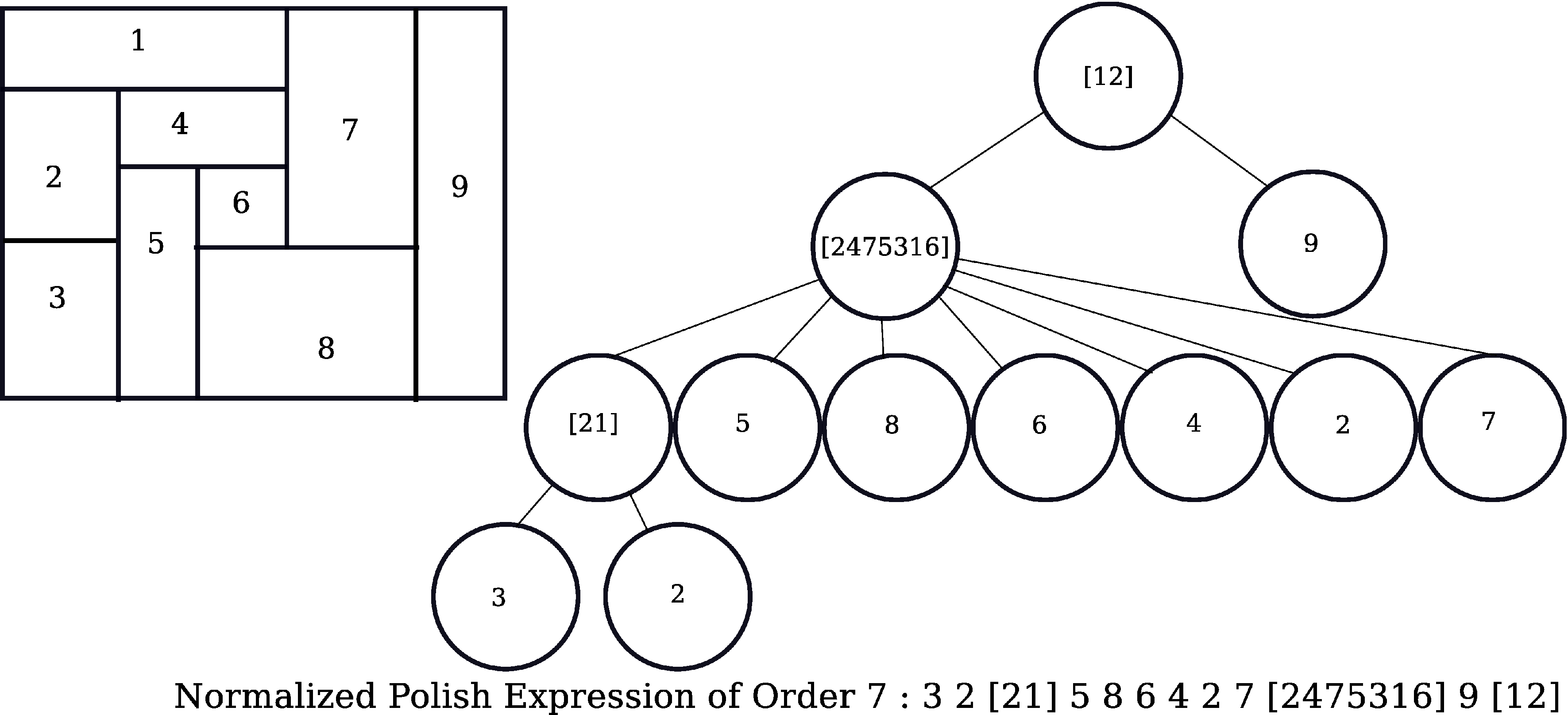}
\end{figure}

Now we will define moves on the Normalized Polish Expressions of order $k$
to define the neighbourhood relationship amongst $\HFO_k$ floorplans.

\begin{enumerate}
\item \textbf{M1}: \textit{Swap two adjacent elements}
  \begin{enumerate}
  \item operand $\longleftrightarrow$ operand

        e.g. $45312[21]2[12]6\mathbf{87}[41352][21] \rightarrow 4531[21]26\mathbf{78} [41352][21]$

        Two elements are \textit{adjacent} if they are adjacent in the
        sequence obtained by removing all operators from the normalized
        $2-5$ expression.
  \item operand $\longleftrightarrow$ operator

        e.g. $45312[21]2[12]68\mathbf{7[41352]} [21] \rightarrow 4531[21]268\mathbf{[41352] 7}[21]$
  \item operator $\longleftrightarrow$ operator 

        e.g. $45312[21]2[12]687 [41352] [21] \rightarrow 4531[21]268 7 [41352] [21]$

        In this case it is required that at most one operator is slicing($[12]/[21]$).
  \end{enumerate}
\item \textbf{M2}: \textit{Complimenting}
  \begin{enumerate}
  \item Complement a maximal chain

        e.g. $67812345[25314] [12] [21] [12] \rightarrow 67812345[25314] [21] [12] [21]$

        A \textbf{maximal chain} is a sequence of  slicing operators $\alpha_i,\alpha_{i+1},\alpha_{i+2},\dots,\alpha_j$ in
        the given normalized expression $\alpha$ such that $\alpha_{i-1}$ and $\alpha_{j+1}$
        if they exists should not be slicing operators.
  \item \textit{Complement an $\HFO_j$ operator}

        In this operation you can take Uniquely $\HFO_j$ operator in the given normalized
        polish expression and replace it with another Uniquely $\HFO_j$ operator.

        e.g. $67812345\mathbf{[41352]} [12] [21] [12] \rightarrow 67812345\mathbf{[25314]} [12] [21] [12]$ 
  \end{enumerate}
\item \textbf{M3}: \textit{Create an $\HFO_j$ operator}

      e.g. $6 7 8 1 4 5 [12] 3 [12] 2 [12] [21] [12] [21] [12] [21] \rightarrow 6 7 8 1 4 5 3 2 [41352] [12] [21] [12]$

      Select a composite rectangle that can be partitioned into $j$ basic/composite 
      rectangles which are not arranged into a Uniquely $\HFO_j$ floorplan and
      re-arrange the $j$ components into an Uniquely $\HFO_j$ floorplan. 
\item \textbf{M4}: \textit{Destroy an $\HFO_j$ operator}

      e.g. $32[21]586427[2475316]9[12]\rightarrow 32[21]586427[41352][12][21]9[12]$
  
      Here we differ slightly from \cite{WongAndThe}. Since $j$ 
      can vary from $2$ to $k$, the replacement policy is uniform.
      We replace a wheel operator like they do. But for destroying
      an $\HFO_j$ operator for $j>5$, we replace it with a wheel
      operator in the beginning and a chain of alternating slicing 
      operators such that the normalization property is not violated.
\end{enumerate}

The moves \textbf{M1}, \textbf{M1(b)} and
\textbf{M2(c)} might produce a sequence that violates condition $2$ in the set
of conditions for normalized polish expressions expression of order $k$. But
here also, checking whether resulting expression is normalized can be done efficiently.  

Given a normalized polish expression of order $k$, it neighbours are all valid
normalized polish expression which can be obtained by a single move from the 
list of moves above. It can be proved that the diameter of the solution
space, that is the maximum distance between two valid normalized polish expressions
of order $k$ of length $n$, is $O(n^2)$. We prove this by observing that within 
$O(n)$-destroy $\HFO_j$ operator moves, all the operators in the given expression
can be made slicing operators. For each operator, with $O(n)$ operand-operator
swap moves, it can be moved to the end of the expression. Hence within $O(n^2)$ 
steps any normalized expression of order $k$ of length $n$ can be transformed into
an expression where all the operands are at the beginning and all the operators 
are at the end, and are slicing operators. The moves are defined such that if 
an expression can be obtained from another using a single move, there exists 
another moves which returns it back to the original. Hence we have proved existence
of a normalized polish expression of order $k$  which is a distance of $O(n^2)$ from
any other expression. Hence between two normalized polish expressions of order $k$,
there is a path of length $O(n^2)$ through this special node.

\section{Recurrence relation for the number $\HFO_5$ floorplans}

Hierarchical Floorplans of Order $5$ is the only $\HFO_k$ other than slicing
floorplans which have been studied in the literature to the best of our knowledge.
Since we have proved that the number of distinct $\HFO_5$ floorplans
with $n$ rooms is equal to the number of distinct skewed generating
trees of order $5$ with $n$ leaves(also proved by \cite{WongAndThe})
it suffices to count such trees .  Let $t_{n}$ denote the number of
distinct skewed generating trees of order $k$ with $n$ leaves and
$t_{1}=1$ representing a tree with a single node. Let $a_{n}$ denote
such trees whose root is labeled $12$, $b_{n}$ denote trees whose root
is labeled $21$, $c_{n}$ denote trees whose root is labeled $41352$
and $d_{n}$ denote the trees whose root is labeled $25314$. Since
these are the only Uniquely $\HFO_k$ permutations for $k\leq5$ the
root has to labeled by one of these. Hence

\[
t_{n}=a_{n}+b_{n}+c_{n}+d_{n}\]

Since it is a skewed tree if the root is labeled $12$, its left child
cannot be $12$ but it can be $12$,$41352$ ,$25314$ or a leaf node.
Similarly if the root is labeled $21$ its left child cannot be $21$
but it can be $12$,$41352$ ,$25314$ or a leaf node.  But for trees
whose roots are labeled $41352/25314$ can have any label for any of
the five children. Hence we get,

\begin{eqnarray*}
a_{n} & = & t_{n-1}.1+\Sigma_{i=2}^{n-1}t_{n-i}(b_{i}+c_{i}+d_{i})\\
b_{n} & = & t_{n-1}.1+\Sigma_{i=2}^{n-1}t_{n-i}(a_{i}+c_{i}+d_{i})\\
c_{n} & = & \Sigma_{\left\{ i,j,k,l,m\geq1|i+j+k+l+m=n\right\} }t_{i}t_{j}t_{k}t_{l}t_{m}\\
d_{n} & = & \Sigma_{\left\{ i,j,k,l,m\geq1|i+j+k+l+m=n\right\} }t_{i}t_{j}t_{k}t_{l}t_{m}\end{eqnarray*}

So $c_{n}=d_{n}$. Also note that since a node labeled $41352/25314$
ought to have five children, $c_{n,}d_{n}=0$ for $n<5$. Summing
up $a_{n}$ and $b_{n}$ and using the identity $t_{i}=a_{i}+b_{i}+c_{i}+d_{i}$
we get

\begin{align*}
a_{n}+b_{n} & =t_{n-1}+t_{n-1}t_{1}+\Sigma_{i=2}^{n-1}t_{n-i}(a_{i}+b_{i}+c_{i}+d_{i}+c_{i}+d_{i})\\
 & =t_{n-1}+\Sigma_{i=1}^{n-1}t_{n-i}t_{i}+2\Sigma_{i=2}^{n-1}t_{n-i}c_{i}\end{align*}

If we substitute for $c_{i}$ in $\Sigma_{g=1}^{n-1}t_{n-g}c_{g}$,
we will get 

\[
\Sigma_{\left\{ h,i,j,k,l,m\geq1|h+i+j+k+l+m=n\right\} }t_{h}t_{i}t_{j}t_{k}t_{l}t_{m}\]
because if you notice the $t_{n-g}$ runs from $1$ to $n-1$ and
$i,j,k,l,m$ in the expansion of $c_{i}$ sums up to $g$, hence if
we let $h=n-g$ then we get $h+i+j+k+l+m=n$. Thus we get the following
recurrence for $t_{n}$

\begin{center}$\begin{array}{ccc}
t_{n} & = & t_{n-1}+\Sigma_{i=1}^{n-1}t_{n-i}t_{i}+\\
 &  & 2\Sigma_{\left\{ h,i,j,k,l,m\geq1|h+i+j+k+l+m=n\right\} }t_{h}t_{i}t_{j}t_{k}t_{l}t_{m}+\\
 &  & 2\Sigma_{\left\{ i,j,k,l,m\geq1|i+j+k+l+m=n\right\} }t_{i}t_{j}t_{k}t_{l}t_{m}\end{array}$\end{center}

We were not able to solve the recurrence using 
the ordinary generating function $T(z)$ 
associated with the sequence $t_n$ defined below.

\[
T(z)=\Sigma_{n=1}^{\infty}t_{n}z^{n-1}
\]

We multiplied the recurrence with $\Sigma_{n=1}^{\infty}z^{n-1}$,
to get,
\[
T(z)=zT(z)+zT^{2}(z)+z^{4}T^{5}(z)+z^{5}T^{6}(z)+t_{1}
\]
Substituting $t_{1}=1$, we get the following polynomial equation
in $T(z)$

\[
z^{5}T^{6}(z)+z^{4}T^{5}(z)+zT^{2}(z)+(z-1)T(z)+1=0
\]
Unfortunately this is a polynomial of sixth degree. Hence
no general solution is available for its roots, which are
needed to obtain the closed form expression for the above
recurrence relation.

Note that in a similar way recurrence relation for any $\HFO_k$
can be constructed by counting the skewed generating trees of order
$k$ where the roots can be any Uniquely $\HFO_l$ permutation for
$l\leq k$.  With our characterization of $\HFO_l$ and the algorithm
for checking whether a permutation is  $\HFO_l$, we can easily find out the number 
of Uniquely $\HFO_l$ permutations for any $l$ and easily get the
recurrence for any $\HFO_k$ using the above mentioned method.

\section{Poly-time Algorithm for counting $\HFO_5$ permutations}

Note that the recurrence obtained above can be used to construct a
polynomial time algorithm for finding $t_{n}$ thus the number of
distinct $\HFO_5$ floorplans with $n$ rooms. We are going to use
dynamic programming to compute the value of $t_{n}$ using the above
recurrence relation. The algorithm is fairly straight forward.

The table $T$ is used to store the values of $t_i,1\leq i \leq n$.
The \textbf{for loop} of lines 3-29, computes successive values of
$t_i$ using the recurrence relation we obtained earlier.
The algorithm runs in time $O(n^{6})$. In general the algorithm for $\HFO_k$
based on a recurrence obtained using the above method will run in
time $O(n^{k+1})$.

\begin{algorithm}

\SetKwInOut{Input}{Input}
\SetKwInOut{Output}{Output}

\SetKwData{MVar}{m}
\SetKwData{IVar}{i}
\SetKwData{JVar}{j}
\SetKwData{KVar}{k}
\SetKwData{LVar}{l}
\SetKwData{HVar}{h}
\SetKwData{XVar}{x}
\SetKwData{YVar}{y}
\SetKwData{ZVar}{z}
\SetKwData{TVar}{T}

\SetKwFunction{Array}{Array}

\SetKw{New}{new}

\SetKwFunction{Min}{min}

\TVar$\leftarrow$\New \Array{n}\;
\TVar$[1]\leftarrow 1$\;

\For{\MVar$=2$ \KwTo $n$} {
  \XVar$\leftarrow$0,\YVar$\leftarrow$0,\ZVar$\leftarrow$0\;
  \For{\IVar$=1$ \KwTo \MVar$-1$}{
    \XVar$\leftarrow$\XVar+\TVar$[\IVar][\MVar-\IVar]$\;
  }
  \For{$\IVar=1$ \KwTo $\MVar-4$}{
    \For{$\JVar=1$ \KwTo \Min{{\MVar-\IVar},{\MVar$-4$}}}{
      \For{$\KVar=1$ \KwTo \Min{{\MVar-$(\IVar+\JVar)$},{\MVar$-4$}}}{
        \For{$\LVar=1$ \KwTo \Min{{\MVar-$(\IVar+\JVar+\KVar)$},{\MVar-$4$}} }{
          $\YVar \leftarrow \YVar + \TVar[\IVar]*\TVar[\JVar]*\TVar[\KVar]*\TVar[\LVar]*\TVar[\MVar-(\IVar+\JVar+\KVar+\LVar)]$\;
        }
      }
    }
  }
 \For{$\HVar=1$ \KwTo $\MVar-5$}{
   \For{$\IVar=1$ \KwTo \Min{{$\MVar-\HVar$},{$\MVar-5$}} } {
     \For{$\JVar=1$ \KwTo \Min{{$\MVar-(\HVar+\IVar)$},{$\MVar-5$}} } {
       \For{$\KVar=1$ \KwTo \Min{{$\MVar-(\HVar+\IVar+\JVar)$},{$\MVar-5$}} }{
         \For{$\LVar=1$ \KwTo \Min{{$\MVar-(\HVar+\IVar+\JVar+\KVar$},{$\MVar-5$}} } {
           $\ZVar \leftarrow \ZVar + \TVar[\HVar]*\TVar[\IVar]*\TVar[\JVar]*\TVar[\KVar]*\TVar[\LVar]*\TVar[\MVar-(\HVar+\IVar+\JVar+\KVar+\LVar)]$\;
         }
       }
     }
   }
 }
 \TVar$[\MVar] \leftarrow$ \XVar + 2\YVar + 2\ZVar + \TVar$[\MVar-1]$\;
}

\textbf{Output} \TVar$[n]$\;

\caption{Algorithm for producing the count of number of distinct $\HFO_5$ floorplans of $n$ rooms}
\end{algorithm}

\section{Properties of Baxter Permutations}

Baxter permutations are an interesting family of permutations combinatorially. They
were first introduced to solve a conjecture about fixed points of commutative functions
by . They are interesting from the VLSI perspective because of their bijective correspondence
to mosaic floorplans. In this chapter we explore some properties of Baxter permutation which
can be easily associated with properties of corresponding mosaic floorplans. The first such
property is that Baxter permutations are closed under inverse. We give a direct
proof for this by the method of contradiction. Then we prove that the mosaic floorplan 
corresponding to the inverse is obtained by taking a mirror image of the floorplan 
corresponding to the permutation about the horizontal axis.
Another such result is that reverse of a Baxter
permutation is also a Baxter permutation. This is a straight forward observation from the 
definition of Baxter permutations itself. But this result becomes interesting when the 
connection to geometry is made. 
The geometric operation on a mosaic floorplan corresponding to reverse on a 
Baxter permutation, is rotating the mosaic floorplan first by $90^{\circ}$s clockwise
and then taking a mirror image along the horizontal axis.

\subsection{Closure Under Inverse}
\begin{thm}
If a permutation $\pi\in S_{n}$ is Baxter then so is $\pi^{-1}$.
\end{thm}

We prove this by giving a direct prove using the method of contradiction.

\begin{proof}Suppose it is not, then there is a text matching $3142/2413$
with absolute difference of first and last being exactly one. Let
this text be at locations $i,j,k,l$ with $i<j<k<l$. Suppose $\pi^{-1}[i],\pi^{-1}[j],\pi^{-1}[k],\pi^{-1}[l]$
forms the pattern $2413$, the we know that $\pi^{-1}[i]+1=\pi^{-1}[l]$
and $\pi^{-1}[k]<\pi^{-1}[i]<\pi^{-1}[l]<\pi^{-1}[j]$. Hence $\left\{ i,j,k,l\right\} $
appears in the order $(k,i,l,j)$ in $\pi$ with $i,l$ appearing
in consecutive locations so they form the pattern $3142$. If $k=j+1$
then this violates the assumption that $\pi$ is Baxter. If $k>j+1$
then $j+1$ has to appear before $i$ or after $l$ in $\pi$ as $i$
and $l$ appear in consecutive positions. If $j+1$ appears before
$i$ in $\pi$ then $j+1,i,l,j$ forms the pattern $3142$ with absolute
difference of first and last being one thus violating the assumption
that $\pi$ is Baxter. So the only place $j+1$ could be is after
$l$, now consider $k,i,l,j+1$, this forms the pattern $3142$ and
$|k-(j+1)|<|k-j|$, so if still $|k-(j+1)|>1$ then we could apply
the same argument as above and include $j+2$. This process cannot
go on for ever as each time $|k-(j+i)|$ is decreasing in value. So
after $|k-j|-1$ steps you will get a text matching the pattern $3142$
with absolute difference of first and last being one thus contradicting
the assumption that $\pi$ is Baxter. Since we have exhausted all
the cases and arrived at a contradiction in each one our assumption
that $\pi^{-1}$ contained a text matching $2413$ with absolute difference
of first and last being one is wrong. Similarly it can be proved that
$\pi^{-1}$ does not contain any text matching $3142$ with absolute
difference of first and last being one. Hence the theorem. \end{proof}

We know prove that equivalent operation on a mosaic floorplan corresponding
to the inverse, is taking the mirror image about vertical axis.
\begin{thm}
  Let $f_{\pi}$ denote the mosaic floorplan corresponding to a Baxter
  permutation. For any given Baxter permutation $\pi$, the floorplan
  corresponding to inverse, $f_{\pi^{-1}}$ can be obtained from $f_{\pi}$
  by taking a mirror image about the horizontal axis.
\end{thm}
\begin{proof}
Let $\pi$ be a Baxter permutation of length $n$. Let us take two indices
$i$ and $j$ such that $i<j$. Consider $\pi[i]$ and $\pi[j]$, either
$\pi[i] < \pi[j]$ or $\pi[i] > \pi[j]$. 

\textbf{Case I:$\pi[i] < \pi[j]$}

Since $\pi[i] < \pi[j]$ and $\pi[i]$ appears before $\pi[j]$, $\pi[i]$ is to the left of $\pi[j]$
in the mosaic floorplan corresponding to $\pi$, denoted by $f_{\pi}$.
In the inverse of $\pi$, $\pi^{-1}$ indices $\pi[i]$ and $\pi[j]$ will
be mapped to $i$ and $j$ respectively. Hence in $f_{\pi^{-1}}$, the basic
rectangles labeled by $i$ and $j$ will be such that $i$ precedes $j$ 
in the top-left deletion ordering(as $i<j$) and also in bottom
left deletion ordering(as $\pi[i]<\pi[j]$). Hence $i$ is to the left of $j$
in $f_{\pi^{-1}}$.

\textbf{Case II:$\pi[i] > \pi[j]$}

Since $\pi[i] > \pi[j]$ and $\pi[i]$ appears before $\pi[j]$ , $\pi[i]$ is below $\pi[j]$
in $f_{\pi}$.
In the inverse of $\pi$, $\pi^{-1}$ indices $\pi[i]$ and $\pi[j]$ will
be mapped to $i$ and $j$ respectively. Hence in $f_{\pi^{-1}}$, the basic
rectangles labeled by $i$ and $j$ will be such that $i$ precedes $j$ 
in the top-left deletion ordering(as $i<j$) but in bottom
left deletion ordering $j$ precedes $i$(as $\pi[i]<\pi[j]$). Hence $i$ is above $j$
in $f_{\pi^{-1}}$.

Hence we get a mapping between the basic rectangles of $f_{\pi}$ and
basic rectangles of $f_{\pi^{-1}}$, such that whenever $\pi[i]$ is
below $\pi[j]$, their images $i,j$ will be such that $i$ will be above
$j$.  And whenever $\pi[i]$ is to the left of $\pi[j]$ so is $i$ and
$j$. The geometrical operation which flips the above/below relation
but does not affect the left/right relation is flipping the object
about the horizontal axis. Hence the theorem. Figure
\ref{fig:GeomEquivInverse} illustrates the above mentioned link
between inverse and the geometry.
\begin{figure}
  \centering
  
  \caption{Obtaining a mosaic floorplan corresponding to the inverse of a Baxter permutation}
  \label{fig:GeomEquivInverse}
  
  \includegraphics[scale=0.4]{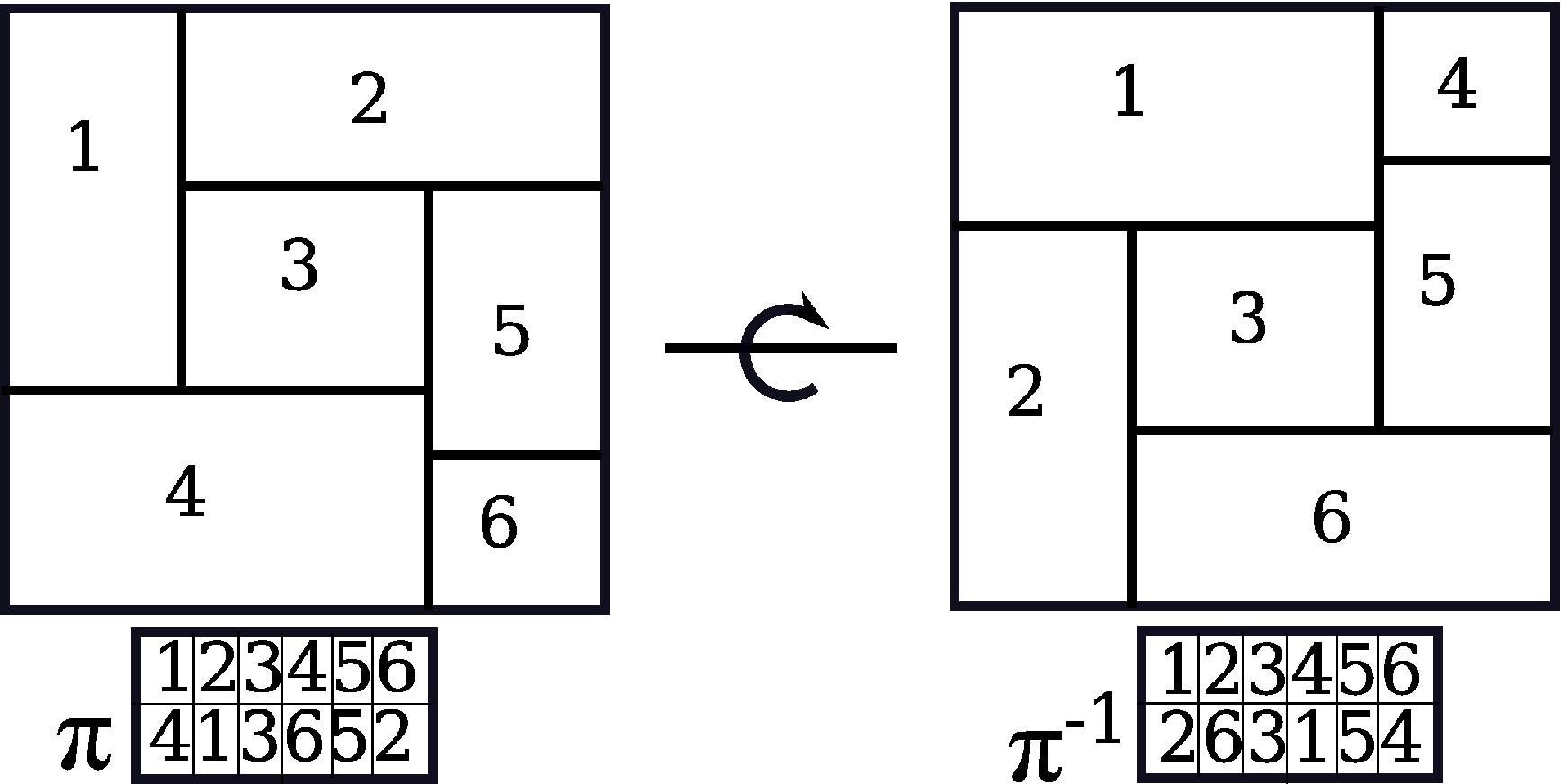}
\end{figure}
\end{proof}

\subsection{Closure under reverse}

\begin{thm}
  If $\pi$ is a Baxter permutation then so is its reverse, more over
  the mosaic floorplan corresponding to reverse of $\pi$ can be
  obtained from mosaic floorplan corresponding to $\pi$ by applying
  clockwise rotation by $90^{\circ}$s and then applying reflection 
  about the vertical axis.
\end{thm}

\begin{proof}
  By definition Baxter permutations itself it is clear that the
  reverse of a Baxter permutation is also a Baxter permutation. 
  But let us find out what is the equivalent operation on the
  mosaic floorplan corresponding to the Baxter permutation which
  produces the mosaic floorplan corresponding to the reverse
  of the given Baxter permutation. Let $\pi$ be a Baxter permutation
  and as above let $f_{\pi}$ represent the mosaic floorplan corresponding
  to $\pi$. Let us take two indices $i$ and $j$ in $\pi$ such that $i<j$.
  If $\pi[i] < \pi[j]$ in $\pi$ we know from the above analysis that $\pi[i]$ is 
  to the left of $\pi[j]$ in $f_{\pi}$. Let $\pi^{r}$ represent the reverse
  of the permutation $\pi$ and let $f_{\pi^{r}}$ represent the mosaic
  floorplan corresponding to the reverse of $\pi$, $\pi^{r}$. In $\pi^{r}$
  the order of numbers $\pi[i]$ and $\pi[j]$ will be reversed, and they
  will be at locations $k=n-i+1$ and $l=n-j+1$ respectively. 
  That is $l < k$ and $\pi^{r}[l] > \pi^{r}[k]$ $\pi[j] > \pi[i]$ we get that in $f_{\pi^{r}}$ 
  $\pi[j]$ is below $\pi[i]$. If $\pi[i] > \pi[j]$ in $\pi$ then we know from the
  proof of the earlier theorem that $\pi[i]$ is below $\pi[j]$.
  In the reverse $\pi^{r}$ the order of numbers $\pi[i]$ and $\pi[j]$ will 
  be reversed. Hence they will be at locations $k=n-i+1$ and 
  $l=n-j+1$ respectively. That is $l < k$ and $\pi^{r}[l] < \pi^{r}[k]$ as
  $\pi[j] < \pi[i]$. This implies that $\pi[j]$ is to the left of $\pi[i]$
  in $f_{\pi^{r}}$. Summarizing this, if the room labeled
  $\pi[i]$ is to the left of $\pi[j]$ in $f_{\pi}$ then 
  the operation which obtains the floorplan corresponding to
  the reverse of $\pi$ will change the relative 
  ordering of these blocks such that $\pi[i]$ will be below
  $\pi[j]$. And similarly if the room labeled $\pi[i]$ is 
  below $\pi[j]$ in $f_{\pi}$ then the operation corresponding
  to the reverse will change their relative ordering such that $\pi[i]$
  is to the left of $\pi[j]$. This corresponds to the rotation by $90^{\circ}$
  clock-wise and then taking mirror image along the horizontal axis.
  Figure
  \ref{fig:GeomEquivReverse} illustrates the above mentioned link
  between reverse and the geometry.

  \begin{figure}
    \centering
    
    \caption{Obtaining a mosaic floorplan corresponding to the reverse of a Baxter permutation}
    \label{fig:GeomEquivReverse}
    
    \includegraphics[scale=0.4]{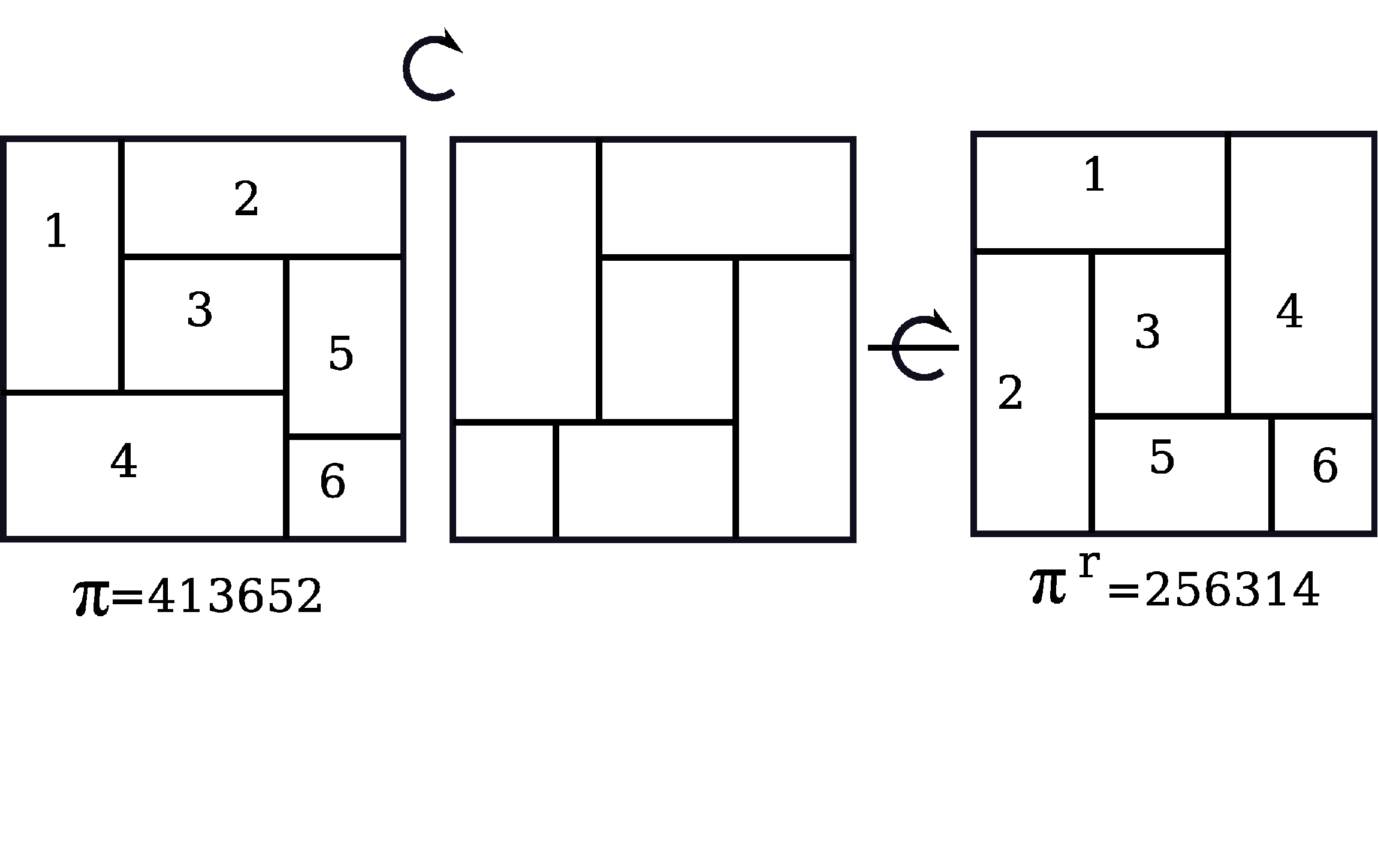}
  \end{figure}
\end{proof}

%%%%%%%%%%%%%%%%%%%%%%%%%%%%%%%%%%%%%%%%%%%%%%%%%%%%%%%%%%%%
% Discussions and Open Questions

\section{Summary}

We characterized permutations corresponding to the Abe-label of $\HFO_k$
floorplans. We also proved that $\HFO_k$ floorplans are in bijective
correspondence with skewed generating trees of Order $k$. This gave
us a recurrence relation for the exact number of $\HFO_k$ floorplans
with $n$ rooms and thus a polynomial time algorithm for generating
the count for any given $n$. We obtained a linear time algorithm
for checking if a given permutation is $\HFO_k$ for a particular
value of $k$. The same algorithm can be used to check whether a permutation
is $\HFO_k$ for some unknown $k$ in $O(n^{2}\log n)$ time. 
We extended the neighbourhood moves on $\HFO_k$ floorplans 
for stochastic search methods like simulated annealing on
these family of floorplans.
We also
proved that Baxter permutations are closed under inverse and
reverse.

Even though we were able to obtain a recurrence relation for the exact number
of $\HFO_k$ floorplans with $n$ rooms and thus a polynomial time
algorithm for generating the count for any given $n$, we were not able
to find a closed form expression for the number of distinct $\HFO_k$
floorplans with $n$ rooms. Even for a particular value of $k$
(especially $5$), it would be interesting to see a closed form
expression for the number of distinct $\HFO_k$ floorplans. Another
open question arising from our research is the number
of distinct Uniquely $\HFO_k$ floorplans. We were able to obtain some trivial
lower bounds based on the construction method described in 
the proof of Infinite hierarchy. But no closed form expression
for the number of Uniquely $\HFO_k$ floorplans were obtained.

\bibliographystyle{plain}
\bibliography{/home/sajin/LaTeX/Thesis/references}

\end{document}